\documentclass[11pt,a4paper]{article}
\usepackage[english]{babel}
\usepackage{amsmath,amsthm,amssymb,epsfig,latexsym}
\usepackage[colorlinks=true,linkcolor=blue,citecolor=magenta]{hyperref}
\usepackage{color,ulem,bm,tikz}
\usepackage[numbers,square,sort&compress]{natbib}

\newcommand{\be}[1]{\begin{equation}\label{#1}}
\newcommand{\ba}[1]{\begin{multline}\label{#1}}
\newcommand{\ee}{\end{equation}}
\newcommand{\ea}{\end{multline}}

\newcommand{\sgn}{\mathop{\rm sgn}}

\newtheorem{prop}{Proposition}[section]
\newtheorem{lemma}{Lemma}[section]

\newtheorem{cor}{Corollary}[section]
\newtheorem{Def}{Definition}[section]

\def\qed{\hfill\nobreak\hbox{$\square$}\par\medbreak}
 \makeatletter
 \@addtoreset{equation}{section}
 \makeatother
 
\newcommand{\bea}{\begin{eqnarray}}
\newcommand{\eea}{\end{eqnarray}}

\newcommand{\bv}{\bar v}

\newcommand{\bt}{\bar t}

\def\CC{\mathbb{C}}
\def\ZZ{\mathbb{Z}}

\def\r#1{\eqref{#1}}
\def\sk#1{\left(#1\right)}

\def\ggo{\mathfrak{g}}
\def\hgo{\mathfrak{h}}

\def\sF{\mathsf{F}}

\def\PSR{\Gamma}
\def\rt{\mathsf{r}}
\def\frd{\xi}
\def\Ig{I_\ggo}

\def\cre{\eta}
\def\ocre{\overline{\eta}}

\def\gga{\mathfrak{sl}_n}
\def\goN{\mathfrak{o}_N}
\def\gsp{\mathfrak{sp}_{2n}}
\def\ggc{\mathfrak{sp}_{2n}}
\def\ggb{\mathfrak{o}_{2n+1}}
\def\ggd{\mathfrak{o}_{2n}}

\def\ago{\mathfrak{a}}

\def\fg{\mathfrak{g}}
\def\fa{\mathfrak{a}}
\def\fb{\mathfrak{b}}
\def\fm{\mathfrak{m}}
\def\ff{\mathfrak{f}}
\def\fe{\mathfrak{e}}

\def\fl{\mathfrak{l}}
\def\fs{\mathfrak{s}}
\def\fp{\mathfrak{p}}
\def\fo{\mathfrak{o}}

\def\gl#1{\fg\fl_{#1}}

\def\cV{\mathcal{V}}
\def\cD{\mathcal{D}}

\def\yF{\mathcal{F}}
\def\yE{\mathcal{E}}
\def\yH{\mathcal{H}}
\def\aF{{F}}
\def\aE{{E}}
\def\aH{{H}}

\def\yG{\mathcal{G}}
\def\scp{\fb}

\def\aY#1{Y(#1)}
\def\sY#1{Y_{\mathrm{D}}(#1)}
\def\aDY#1{\cD Y(#1)}
\def\sDY#1{\cD Y_{\mathrm{D}}(#1)}

\def\aDX#1{\cD X(#1)}
\def\sDX#1{\cD X_{\mathrm{D}}(#1)}

\def\frd{\xi}
\def\epsi{\epsilon}
\def\bvphi{\varphi}
\def\ups{\bm{\omega}}
\def\hc{\mathsf{h}}
\def\fc{\mathsf{f}}

\def\ftr{\mathbf{f}}

\def\Ccc{\mathcal{C}}
\def\oY{\overline{Y}}
\def\cX{\mathcal{X}}
\def\cY{\mathcal{Y}}
\def\bfC{\bm{C}}

\begin{document}

\vspace{12pt}

\begin{center}
\begin{LARGE}
{\bf Serre Relations  in Yangian Doubles}
\end{LARGE}

\vspace{10mm}

\begin{large}
A.~Liashyk${}^{a}$, S.~Pakuliak${}^{b}$ and E.~Ragoucy${}^{b}$
\end{large}

\vspace{10mm}

${}^a$ {\it Beijing Institute of Mathematical Sciences and Applications (BIMSA),\\
No. 544, Hefangkou Village Huaibei Town, Huairou District Beijing 101408, China}

\vspace{2mm}

${}^b$ {\it Laboratoire d'Annecy de Physique Théorique (LAPTh)\\ 
CNRS \& Université Savoie Mont Blanc\\
Chemin de Bellevue, BP 110, F-74941, Annecy-le-Vieux cedex, France}
\vspace{2mm}

E-mails: liashyk@bimsa.cn, pakuliak@lapth.cnrs.fr, ragoucy@lapth.cnrs.fr

\bigskip

\end{center}

\begin{abstract}
Following the approach of B.~Enriquez~\cite{E} we 
exhibit the analytical properties of the products of the currents 
 in the Yangian doubles 
restricted to the category of the highest weight representations. 
We will demonstrate that the Serre relations for the simple root currents in 
the Drinfeld's 'new' 
realization of the Yangian doubles \cite{Dnew,KhT-DY,LP1} can be reformulated 
as  quadratic commutation relations between  composed 
currents for the Yangian doubles associated with Lie algebras of the 
classical series. 
\end{abstract}


\section{Introduction}

The Yangian algebra $\sY{\fg}$ for any simple Lie algebra $\fg$ was introduced 
by V.~Drinfeld in his fundamental paper \cite{D85}. 
It is defined as a deformation of the universal enveloping algebra 
$U(\tilde\fg[z])$ for the polynomial loop Lie algebra $\tilde\fg[z]$ and 
is related to the rational solutions of  quantum Yang-Baxter 
equation. If the commutation relations for the polynomial currents 
in the loop Lie algebra $\tilde\fg[z]$ are defined using a classical solution 
of the Yang-Baxter equation, then  the commutation relations 
in the Yangian $\sY{\fg}$ are defined using the corresponding 
solution of the quantum Yang-Baxter equation.  The algebra $\sY{\fg}$
was called Yangian in \cite{D85} in honor of C.~N.~Yang who found 
the first solution of the quantum Yang-Baxter equation in  \cite{Yang}.  
  
In \cite{D85} the Yangian $\sY{\fg}$ was defined using a finite set 
of generators. Lately, in \cite{Dnew} a 'new' realization of the Yangian 
was presented which uses a countable set of generators. This countable 
set of generators can be viewed as a deformation of the generators  
 of the polynomial loop algebra $\tilde\fg[z]$.  The commutation relations 
 for these generators of $\sY{\fg}$ can be found in the pioneering papers 
 \cite{Dnew,D-FTINT} which we recast in the definition~\ref{sYn}. 

In the present paper, we will work with a slightly different realization 
of the Yangian which we denoted as $\aY{\fg}$ and which is isomorphic 
as universal enveloping algebra to the Yangian $\sY{\fg}$ in its 'new' 
realization \cite{Dnew}. More exactly, we will work with 
the algebra $\aDY{\fg}$ which is isomorphic as an associative algebra to 
a quantum double $\sDY{\fg}$ of the Yangian $\sY{\fg}$. As a Hopf algebra, the
Yangian double $\aDY{\fg}$ may also be constructed as a quantum double
of the Yangian $\aY{\fg}$ using the coalgebraic structure of the latter algebra 
and following calculations presented in \cite{KhT-DY}.

In this paper, we investigate the algebraic structure of $\aDY{\fg}$ as defined by relations among generating series (currents), restricting our attention to the category of highest-weight representations.
Among these relations, of particular interest  
 are the so-called Serre relations for the currents
associated to the simple roots of the finite dimensional Lie algebra $\fg$. 
The goal of this paper is twofold. First, 
following ideas of the paper \cite{E} we will use the Serre relations 
to describe certain  analytical properties of the product of the currents in 
$\aDY{\fg}$ restricted to the category of the highest weight representations. 
Second, for the Yangian doubles $\aDY{\fg}$ associated 
to simple Lie algebras $\fg$  we will reinterpret 
the Serre relations as 
quadratic commutation relations between   {\sl composed 
 currents} for the Yangian double. Composed currents were introduced
 in \cite{DKh} for the simple laced quantum affine algebra 
 $U_q(\widehat{\ggo})$. 
 The composed currents for the Yangian doubles of the classical 
 series were listed in \cite{LP1}. It was shown there that they are related to 
 the Gaussian coordinates of the first fundamental $T$-matrices of these 
 Yangian doubles considered in their $RTT$ realization. A similar classification 
 of the composed currents for the Yangian doubles associated to 
 exceptional Lie algebras requires the study of their $RTT$ realizations using 
 the corresponding $\fg$-invariant $R$-matrices. 
 We will  address this problem in future research.

 Although currents $\aF_a(u)=\sum_{\ell\in\ZZ}\aF_a[\ell](u/c)^{-\ell-1}$ 
  of the Yangian double $\aDY{\fg}$
 corresponding to the simple roots of the underlying Lie algebra $\fg$
 are formal series depending on a
  formal parameter $u$, one can assign an analytical meaning to 
  their products. Each current 
  in such a product should be replaced by a 
  sum of  'positive' half-current analytical at infinity and 
  'negative' half-current analytical at zero. Then the commutation relations 
  between generators of the Yangian double $\aF_a[\ell]$ 
   allows to normal order any  product 
  of half-currents and present it as the sum of monomials such that all 'hegative' 
half-currents are on the left of all 'positive' half-currents  
(see appendix~\ref{AppA} for examples of such normal ordering).

This normal ordering procedure is equivalent to restrict 
 the Yangian double to the category of its highest weight 
representations as it was explained   in \cite{E} for the 
 quantum affine algebras.  
 The analytical properties of the product of the currents can be 
 extracted from the analytical properties of 
 arbitrary matrix coefficients 
 of  products of currents calculated between vectors from the highest and 
 dual highest weights representations of the Yangian double.
 
 Hence, if 
 one restricts the Yangian double $\aDY{\fg}$ to the category of its highest weight
  representations, then any product of the currents may be treated 
 as an analytical function of their parameters. This will be the ideology that
 we will follow in this paper.  The main objective is 
 an investigation of the properties of the products of  currents 
 which follows from the commutation and the Serre relations
 in the Yangian doubles  restricted to the category of
 the highest weight representations.

 The paper is composed as follows. In  section~\ref{g-all} we introduce notations and define the algebraic structure of the Yangian double $\aDY{\fg}$.
 We detail the difference between Yangians $\aY{\fg}$ and 
 $\sY{\fg}$ and their quantum doubles and describe their isomorphism 
 as universal enveloping algebras. 
 The final part of this section focusses on the analytical properties of products of generating series restricted to the category of highest weight representations.
 The section~\ref{an-pro-BV} is devoted to the extension of  the results 
 of \cite{E} to the Yangian double $\aDY{\fg}$ and relates the Serre relations 
 with certain analytical properties of the products of the currents. 
 Having in mind an application to the quantum integrable models 
 we consider in  section~\ref{CC-SR} the composed currents 
 for Yangian doubles of the classical series and demonstrate that 
 their commutation relations are implied by the corresponding Serre relations.  
 In concluding section~\ref{disc} we shortly 
  discuss an application of the obtained results to the construction of 
  the off-shell Bethe vectors in the framework of 
  algebraic Bethe ansatz for a generic $\fg$-invariant integrable model. 
  The properties of the composed currents
  as  well-defined objects in the completed subalgebra of $\aDY{\fg}$ given by 
  definition~\ref{Comp}   are illustrated for $\fg=\gl{3}$  in appendix~\ref{AppA}. 
 An equivalence of the fifth order Serre relation in the Yangian double 
 $\sDY{\fg_2}$ (associated with the exceptional Lie algebra $\fg_2$) 
 to the quadratic commutation relation between simple root 
 and particular composed current is shown in appendix~\ref{AppB}. 
Finally, in the appendix \ref{AppC}
we described the rational analogs of the
$\delta$-functions identities extending 
some of the results of \cite{E} to the 
Yangian doubles.

\section{Yangian double}\label{g-all}

Let $\mathfrak{g}$ be one of the simple Lie 
algebras $\mathfrak{sl}_n$, $\mathfrak{o}_{2n+1}$, 
$\mathfrak{sp}_{2n}$ and $\mathfrak{o}_{2n}$ 
corresponding to the $A_{n-1}$, $B_n$, $C_n$ and $D_n$  classical series 
or Lie algebra  $\fe_6$, $\fe_7$, $\fe_8$,  $\ff_4$, and $\fg_2$,
corresponding to the exceptional cases  $E_6$, $E_7$, 
$E_8$,  $F_4$ and $G_2$.

Let $\PSR_\ggo$ be the set of labels of the Dynkin diagram for $\ggo$
 which ennumerate the simple roots $\rt_a$, $a\in\PSR_\fg$ of the algebra $\fg$.
We use numbering of the simple roots shown on the following pictures.

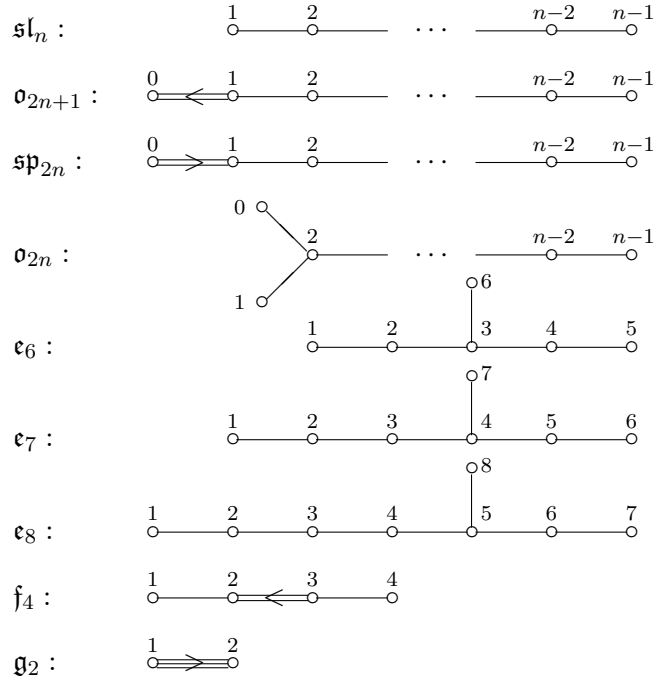
\begin{center}
\begin{picture}(100,15)
\put(-40,0){$\gga:$}
\put(40,0){$\circ$}\put(40.5,7){$\scriptstyle 1$}\put(44,2.5){\line(1,0){27}}
\put(70,0){$\circ$}\put(70.5,7){$\scriptstyle 2$}\put(74,2.5){\line(1,0){27}}
\put(111,0){$\cdots$}
\put(134,2.5){\line(1,0){27}}\put(160,0){$\circ$}\put(155.5,7){$\scriptstyle n-2$}
\put(164,2.5){\line(1,0){27}}\put(190,0){$\circ$}\put(185.5,7){$\scriptstyle n-1$}
\end{picture}
\end{center}

\begin{center}
\begin{picture}(100,15)
\put(-40,0){$\ggb:$}
\put(10,0){$\circ$}\put(10.5,7){$\scriptstyle 0$}
\put(14,3.5){\line(1,0){27}}\put(14,1.5){\line(1,0){27}}\put(24,-0.3){$<$}
\put(40,0){$\circ$}\put(40.5,7){$\scriptstyle 1$}\put(44,2.5){\line(1,0){27}}
\put(70,0){$\circ$}\put(70.5,7){$\scriptstyle 2$}\put(74,2.5){\line(1,0){27}}
\put(111,0){$\cdots$}
\put(134,2.5){\line(1,0){27}}\put(160,0){$\circ$}\put(155.5,7){$\scriptstyle n-2$}
\put(164,2.5){\line(1,0){27}}\put(190,0){$\circ$}\put(185.5,7){$\scriptstyle n-1$}
\end{picture}
\end{center}

\begin{center}
\begin{picture}(100,15)
\put(-40,0){$\gsp:$}
\put(10,0){$\circ$}\put(10.5,7){$\scriptstyle 0$}
\put(14,3.5){\line(1,0){27}}\put(14,1.5){\line(1,0){27}}\put(24,-0.3){$>$}
\put(40,0){$\circ$}\put(40.5,7){$\scriptstyle 1$}\put(44,2.5){\line(1,0){27}}
\put(70,0){$\circ$}\put(70.5,7){$\scriptstyle 2$}\put(74,2.5){\line(1,0){27}}
\put(111,0){$\cdots$}
\put(134,2.5){\line(1,0){27}}\put(160,0){$\circ$}\put(155.5,7){$\scriptstyle n-2$}
\put(164,2.5){\line(1,0){27}}\put(190,0){$\circ$}\put(185.5,7){$\scriptstyle n-1$}
\end{picture}
\end{center}

\begin{center}
\begin{picture}(100,25)
\put(-40,0){$\ggd:$}
\put(55.5,19){\line(1,-1){15}}\put(51,18){$\circ$}\put(43.15,18){$\scriptstyle 0$}
\put(55.5,-14){\line(1,1){15,7}}\put(51,-18){$\circ$}\put(43.15,-18){$\scriptstyle 1$}
\put(70,0){$\circ$}\put(70.5,7){$\scriptstyle 2$}\put(74,2.5){\line(1,0){27}}
\put(111,0){$\cdots$}
\put(134,2.5){\line(1,0){27}}\put(160,0){$\circ$}\put(155.5,7){$\scriptstyle n-2$}
\put(164,2.5){\line(1,0){27}}\put(190,0){$\circ$}\put(185.5,7){$\scriptstyle n-1$}
\end{picture}
\end{center}

\begin{center}
\begin{picture}(100,25)
\put(-40,0){$\fe_{6}:$}
\put(70,0){$\circ$}\put(70.5,7){$\scriptstyle 1$}\put(74,2.5){\line(1,0){27}}
\put(100,0){$\circ$}\put(100.5,7){$\scriptstyle 2$}\put(104,2.5){\line(1,0){27}}
\put(130,0){$\circ$}\put(132.5,4.2){\line(0,1){20}}\put(130,24){$\circ$}
\put(136,7){$\scriptstyle 3$}\put(136,25){$\scriptstyle 6$}
\put(134,2.5){\line(1,0){27}}\put(160,0){$\circ$}\put(160.5,7){$\scriptstyle 4$}
\put(164,2.5){\line(1,0){27}}\put(190,0){$\circ$}\put(190.5,7){$\scriptstyle 5$}
\end{picture}
\end{center}

\begin{center}
\begin{picture}(100,25)
\put(-40,0){$\fe_{7}:$}
\put(40,0){$\circ$}\put(40.5,7){$\scriptstyle 1$}\put(44,2.5){\line(1,0){27}}
\put(70,0){$\circ$}\put(70.5,7){$\scriptstyle 2$}\put(74,2.5){\line(1,0){27}}
\put(100,0){$\circ$}\put(100.5,7){$\scriptstyle 3$}\put(104,2.5){\line(1,0){27}}
\put(130,0){$\circ$}\put(132.5,4.2){\line(0,1){20}}\put(130,24){$\circ$}
\put(136,7){$\scriptstyle 4$}\put(136,25){$\scriptstyle 7$}
\put(134,2.5){\line(1,0){27}}\put(160,0){$\circ$}\put(160.5,7){$\scriptstyle 5$}
\put(164,2.5){\line(1,0){27}}\put(190,0){$\circ$}\put(190.5,7){$\scriptstyle 6$}
\end{picture}
\end{center}

\begin{center}
\begin{picture}(100,25)
\put(-40,0){$\fe_{8}:$}
\put(10,0){$\circ$}\put(10.5,7){$\scriptstyle 1$}\put(14,2.5){\line(1,0){27}}
\put(40,0){$\circ$}\put(40.5,7){$\scriptstyle 2$}\put(44,2.5){\line(1,0){27}}
\put(70,0){$\circ$}\put(70.5,7){$\scriptstyle 3$}\put(74,2.5){\line(1,0){27}}
\put(100,0){$\circ$}\put(100.5,7){$\scriptstyle 4$}\put(104,2.5){\line(1,0){27}}
\put(130,0){$\circ$}\put(132.5,4.2){\line(0,1){20}}\put(130,24){$\circ$}
\put(136,7){$\scriptstyle 5$}\put(136,25){$\scriptstyle 8$}
\put(134,2.5){\line(1,0){27}}\put(160,0){$\circ$}\put(160.5,7){$\scriptstyle 6$}
\put(164,2.5){\line(1,0){27}}\put(190,0){$\circ$}\put(190.5,7){$\scriptstyle 7$}
\end{picture}
\end{center}

\begin{center}
\begin{picture}(100,15)
\put(-40,0){$\ff_4:$}
\put(10,0){$\circ$}\put(10.5,7){$\scriptstyle 1$}\put(14,2.5){\line(1,0){27}}
\put(40,0){$\circ$}\put(40.5,7){$\scriptstyle 2$}
\put(44,3.5){\line(1,0){27}}\put(44,1.5){\line(1,0){27}}\put(54,-0.3){$<$}
\put(70,0){$\circ$}\put(70.5,7){$\scriptstyle 3$}\put(74,2.5){\line(1,0){27}}
\put(100,0){$\circ$}\put(100.5,7){$\scriptstyle 4$}
\end{picture}
\end{center}

\begin{center}
\begin{picture}(100,15)
\put(-40,0){$\fg_2:$}
\put(10,0){$\circ$}\put(10.5,7){$\scriptstyle 1$}
\put(14,4){\line(1,0){27}}\put(14,2.5){\line(1,0){27}}\put(14,1){\line(1,0){27}}
\put(24,0){$>$}
\put(40,0){$\circ$}\put(40.5,7){$\scriptstyle 2$}
\end{picture}
\end{center}

Let $\fb_{a,b}$ and $\ago_{a,b}$ be respectively the symmetric matrix of the 
normalized  simple root scalar products
  and the Cartan matrix of $\ggo$:
\begin{equation}\label{Cmat}
 \fb_{a,b}=(\rt_a,\rt_b)/2=d_a\, \ago_{a,b}/2,\qquad 
d_a=(\rt_a,\rt_a)/2\,.
\end{equation} 

\subsection{Yangians $\aY{\fg}$ and $\sY{\fg}$}
For any finite dimensional Lie algebra $\fg$, the Yangian of $\fg$ possesses several presentations. In the present paper, we will use two of them, that we call the symmetric and the asymmetric presentation. To distinguish these two presentations, we will note the corresponding Yangian $\aY{\fg}$ and $\sY{\fg}$, although the two algebras are isomorphic, see property \ref{asY} below.
The Yangian $\aY{\fg}$ appears when considering the Gauss decomposition of the monodromy matrix of the Yangian \cite{LP1}.
\begin{Def}\label{aYn}
Let 
\begin{equation}\label{Yang-elem}
\aF_a[\ell],\quad \aE_a[\ell],\quad \aH_a[\ell],\quad  \ell\geq 0,\quad 
 a\in\PSR_\ggo
 \end{equation}
be an infinite set of generators satisfying the following commutation relations 
for all $\ell,\ell_1,\ell_2\geq 0$
\begin{subequations}\label{cr-asym}
\begin{equation}\label{cr-m2}
\Big[\aH_a[\ell_1+1],\aF_b[\ell_2]\Big]-\Big[\aH_a[\ell_1],\aF_b[\ell_2+1]\Big]=-\ 
\fb_{a,b}\begin{cases}
2\,\aH_a[\ell_1]\,\aF_b[\ell_2],&a<b\,,\\
2\,\aF_b[\ell_2]\,\aH_a[\ell_1],&a>b\,,\\
\Big\{\aH_a[\ell_1],\aF_b[\ell_2]\Big\},&a=b\,,
\end{cases}
\end{equation}
\begin{equation}\label{cr-m3}
\Big[\aH_a[\ell_1+1],\aE_b[\ell_2]\Big]-\Big[\aH_a[\ell_1],\aE_b[\ell_2+1]\Big]= 
\fb_{a,b}\,\begin{cases}
2\,\aH_a[\ell_1]\,\aE_b[\ell_2],&a>b\,,\\
2\,\aE_b[\ell_2]\,\aH_a[\ell_1],&a<b\,,\\
\Big\{\aH_a[\ell_1],\aE_b[\ell_2]\Big\},&a=b\,,
\end{cases}
\end{equation}
\begin{equation}\label{cr-m5}
\Big[\aF_a[\ell_1+1],\aF_b[\ell_2]\Big]-\Big[\aF_a[\ell_1],\aF_b[\ell_2+1]\Big]=-\ 
\fb_{a,b}\begin{cases}
2\,\aF_a[\ell_1]\,\aF_b[\ell_2],& a<b\,,\\
\Big\{\aF_a[\ell_1],\aF_b[\ell_2]\Big\},&a=b\,,
\end{cases}
\end{equation}
\begin{equation}\label{cr-m7}
\Big[\aE_a[\ell_1+1],\aE_b[\ell_2]\Big]-\Big[\aE_a[\ell_1],\aE_b[\ell_2+1]\Big]= 
\fb_{a,b}\begin{cases}
2\,\aE_b[\ell_2]\,\aE_a[\ell_1],& a<b\,,\\
\Big\{\aE_a[\ell_1],\aE_b[\ell_2]\Big\},& a=b\,,
\end{cases}
\end{equation}
\end{subequations}
\begin{subequations}\label{cr-m}
\begin{equation}\label{cr-m8}
\Big[\aH_a[\ell_1],\aH_b[\ell_2]\Big]=0,
\quad 
\Big[\aE_a[\ell_1],\aF_b[\ell_2]\Big]=\delta_{a,b}\,\aH_a[\ell_1+\ell_2]\,,
\end{equation}
\begin{equation}\label{cr-m1}
\Big[\aH_a[0],\aE_b[\ell]\Big]=2\,\fb_{a,b}\,\aE_b[\ell],\quad 
\Big[\aH_a[0],\aF_b[\ell]\Big]=-\ 2\,\fb_{a,b}\,\aF_b[\ell]\,,
\end{equation}
\begin{equation}\label{Sr-m1}
\mathop{\rm Sym}\limits_{\ell_1,\ldots,\ell_{\fm_{a,b}}}
\Big[\aF_a[\ell_1],\Big[\aF_a[\ell_2],\cdots
\Big[\aF_a[\ell_{\fm_{a,b}}],\aF_b[\ell]\Big]\cdots\Big]\Big]=0\,,
\end{equation}
\begin{equation}\label{Sr-m2}
\mathop{\rm Sym}\limits_{\ell_1,\ldots,\ell_{\fm_{a,b}}}
\Big[\aE_a[\ell_1],\Big[\aE_a[\ell_2],\cdots
\Big[\aE_a[\ell_{\fm_{a,b}}],\aE_b[\ell]\Big]\cdots\Big]\Big]=0\,,
\end{equation}
\end{subequations}  
where $\{A,B\}=A\,B+B\,A$ is the anticommutator and\,\footnote{We recall that $\fa_{a,b}\leq0$ when $a\not=b$.} $\fm_{a,b}=1-\fa_{a,b}$
for $a\not=b$. 
The Yangian $\aY{\fg}$ associated to $\fg$ is an associative algebra 
with a unit element $\mathbf{1}$ generated by the set of generators \r{Yang-elem}. 
Sometimes we will call this realization of Yangian as {\sl asymmetric} due to 
the fact that r.h.s. of equalities \r{cr-asym} are different for $a<b$ and $a>b$. 
\end{Def}

\begin{Def}\label{sYn}
The commutation relations in the Yangian algebra 
 $\sY{\fg}$ in its 'new' realization \cite{Dnew} generated by the 
 elements
\begin{equation}
 \yF_a[\ell],\quad \yE_a[\ell],\quad \yH_a[\ell],\quad \ell\geq0,\quad a\in\PSR_\fg
\end{equation}
are given by the same commutation relations \r{cr-m} while the relations \r{cr-asym} are replaced by the equalities
\begin{subequations}\label{cr-ms}
\begin{equation}\label{cr-ms2}
\Big[\yH_a[\ell_1+1],\yF_b[\ell_2]\Big]-\Big[\yH_a[\ell_1],\yF_b[\ell_2+1]\Big]=-\ 
\fb_{a,b}\,
\Big\{\yH_a[\ell_1],\yF_b[\ell_2]\Big\}\,,
\end{equation}
\begin{equation}\label{cr-ms3}
\Big[\yH_a[\ell_1+1],\yE_b[\ell_2]\Big]-\Big[\yH_a[\ell_1],\yE_b[\ell_2+1]\Big]= 
\fb_{a,b}\,
\Big\{\yH_a[\ell_1],\yE_b[\ell_2]\Big\}\,,
\end{equation}
\begin{equation}\label{cr-ms5}
\Big[\yF_a[\ell_1+1],\yF_b[\ell_2]\Big]-\Big[\yF_a[\ell_1],\yF_b[\ell_2+1]\Big]=-\ 
\fb_{a,b}\,
\Big\{\yF_a[\ell_1],\yF_b[\ell_2]\Big\}\,,
\end{equation}
\begin{equation}\label{cr-ms7}
\Big[\yE_a[\ell_1+1],\yE_b[\ell_2]\Big]-\Big[\yE_a[\ell_1],\yE_b[\ell_2+1]\Big]= 
\fb_{a,b}\,
\Big\{\yE_a[\ell_1],\yE_b[\ell_2]\Big\}
\end{equation}
\end{subequations}  
valid for all values of the indices $a,b\in\PSR_\fg$. 
 We will call this realization of the Yangian as {\sl symmetric} because 
the r.h.s. of equalities \r{cr-ms} are 
the same for $a<b$ and $a>b$.  
\end{Def}

One can rewrite the commutation relation in the Yangian $\aY{\fg}$
in terms of half-currents defined as follows 
\begin{equation}\label{h-cur}
\aF^+_a(u)=\sum_{\ell\geq0}\aF_a[\ell]\Big(\frac{u}{c}\Big)^{-\ell-1},\quad
\aE^+_a(u)=\sum_{\ell\geq0}\aE_a[\ell]\Big(\frac{u}{c}\Big)^{-\ell-1},\quad
\aH^+_a(u)=\mathbf{1}+\sum_{\ell\geq0}\aH_a[\ell]\Big(\frac{u}{c}\Big)^{-\ell-1}
\end{equation}
and analogous formulas for the half-currents $\yF^+_a(u)$,
$\yE^+_a(u)$, $\yH^+_a(u)$ in the Yangian $\sY{\fg}$. 
The commutation relations \r{cr-m} can be written in terms 
of half-currents and for the Yangian $\sY{\fg}$ they can be found in \cite{KhT-DY}. 
For the Yangian $\aY{\fg}$ these commutation relations will be the same 
except those corresponding to \r{cr-asym}. For example, 
for the Yangian 
$\aY{\fg}$ the relation \r{cr-m5} in terms of half-currents takes the form 
for $a<b$
\begin{equation}\label{aFFd}
c^{-1}\,(u-v)\Big[\aF^+_a(u),\aF^+_b(v)\Big]=
\Big[\aF_a[0],\aF^+_b(v)\Big]- \Big[\aF^+_a(u),\aF_b[0]\Big]-2\,\fb_{a,b}\,
\aF^+_a(u)\,\aF^+_b(v)\,.
\end{equation}
Analogously, the commutation relation \r{cr-ms5} in the Yangian 
$\sY{\fg}$ takes the form 
\begin{equation}\label{sFFd1}
c^{-1}\,(u-v)\Big[\yF^+_a(u),\yF^+_b(v)\Big]=
\Big[\yF_a[0],\yF^+_b(v)\Big]- \Big[\yF^+_a(u),\yF_b[0]\Big]-\,\fb_{a,b}\,
\Big\{\yF^+_a(u),\yF^+_b(v)\Big\}
\end{equation}
valid for all values $a,b\in\PSR_\fg$.  The commutators with the zero modes 
can be excluded in the relation \r{sFFd1} and it takes the form presented 
in \cite{KhT-DY}
\begin{equation*}
c^{-1}\,(u-v)\sk{\Big[\yF^+_a(u),\yF^+_b(v)\Big]+\Big[\yF^+_b(u),\yF^+_a(v)\Big]}=
\fb_{a,b}\,\Big\{(\yF^+_a(u)-\yF^+_a(v)),(\yF^+_b(u)-\yF^+_b(v))\Big\}.
\end{equation*}

\begin{prop}\label{asY}
The asymmetric  $\aY{\fg}$ and symmetric $\sY{\fg}$ versions of the Yangian 
algebra are isomorphic as associative algebras. 
\end{prop}
\proof 
We will provide the proof of this statement in the next section by considering 
Yangian doubles $\aDY{\fg}$ ($\sDY{\fg}$) and inclusions of 
$\aY{\fg}\hookrightarrow\aDY{\fg}$ ($\sY{\fg}\hookrightarrow\sDY{\fg}$).
\qed

\subsection{Yangian doubles $\aDY{\fg}$ and $\sDY{\fg}$}

Let $\aDX{\fg}$ and $\sDX{\fg}$ be  algebras generated by the elements 
given  in the definitions~\ref{aYn} and \ref{sYn} for 
 $\ell\in\ZZ$ and satisfying the commutation relations in the corresponding 
 definitions.  Assign to  each 
 generator in the algebras $\aDX{\fg}$ and $\sDX{\fg}$ its degree 
given by the rules 
\begin{equation}\label{degY}
\begin{array}{l}
{\rm deg}\,(\aF_a[\ell])={\rm deg}\,(\aE_a[\ell])={\rm deg}\,(\aH_a[\ell])=\ell,\\[3mm]
{\rm deg}\,(\yF_a[\ell])={\rm deg}\,(\yE_a[\ell])={\rm deg}\,(\yH_a[\ell])=\ell,
\end{array}
\qquad
{\rm deg}(\mathbf{1})=0\,.
\end{equation}
Both algebras $\aDX{\fg}$ and $\sDX{\fg}$ admits a $\ZZ$-filtration 
\begin{equation}\label{filtr}
\begin{array}{c}
\cdots \subset \aDX{\fg}_{-\ell}\subset \cdots\subset  \aDX{\fg}_{-1}\subset
 \aDX{\fg}_{0}\subset\cdots   
 \subset \aDX{\fg}_{\ell}\subset\cdots \subset  \aDX{\fg}\,,\\[3mm]
 \cdots \subset \sDX{\fg}_{-\ell}\subset \cdots\subset  \sDX{\fg}_{-1}\subset
 \sDX{\fg}_{0}\subset\cdots   
 \subset\sDX{\fg}_{\ell}\subset\cdots \subset  \sDX{\fg}
 \end{array}
\end{equation}
defined by the degrees \r{degY}. In \r{filtr} we define $\aDX{\fg}_{\ell}$ 
(resp. $\sDX{\fg}_{\ell}$) as a 
linear span of the elements from $\aDX{\fg}$ (resp. $\sDX{\fg}$)
with degree less or equal to 
 $\ell\in\ZZ$. Let $\overline{\aDX{\fg}}$ (resp. $\overline{\sDX{\fg}}$)
 be the corresponding formal 
completion of $\aDX{\fg}$ (resp. of $\sDX{\fg}$). 
The definition \ref{aYn} of the Yangian $\aY{\fg}$ (resp. definition \ref{sYn} of $\sY{\fg}$)
defines an inclusion 
$\aY{\fg}\hookrightarrow\overline{\aDX{\fg}}$ 
(resp. $\sY{\fg}\hookrightarrow\overline{\sDX{\fg}}$). 

The algebra $\overline{\sDX{\fg}}$ was introduced by V.~Drinfeld 
in \cite{Dnew,D-FTINT} in the framework of the quantum double 
construction using certain coalgebraic structure of the currents associated 
with a Borel subalgebra of $\fg$. The
Yangian double
$\sDY{\fg}$
was obtained as the quantum double \cite{D} of $\sY{\fg}$ in \cite{KhT-DY}.
Following the results of the paper \cite{KhT-DY} 
we assume that the algebras $\sDY{\fg}$ and $\overline{\sDX{\fg}}$
(resp. $\aDY{\fg}$ and $\overline{\aDX{\fg}}$) are isomorphic as associative 
algebras, but they have different coalgebraic structures that we do not discuss in this paper.

We will investigate the algebraic properties of the algebra
$\aDY{\fg}$ presented in the form of the generating series or currents 
\begin{equation}\label{gen-ser-g}
\begin{aligned}
\aF_a(u)&=\sum_{\ell\in\ZZ}\aF_a[\ell](u/c)^{-\ell-1},
\qquad &\aE_a(u)&=\sum_{\ell\in\ZZ}\aE_a[\ell](u/c)^{-\ell-1}\,,\\
\aH_a^+(u)&=\mathbf{1}+\sum_{\ell\geq 0}\aH_a[\ell](u/c)^{-\ell-1},
\qquad
&\aH_a^-(u)&=\mathbf{1}-\sum_{\ell< 0}\aH_a[\ell](u/c)^{-\ell-1}
\end{aligned}
\end{equation}
using a formal parameter $u$.

To describe this presentation 
we introduce the rational functions $g(u,v)$, $h_{a}(u,v)$, $h_{a,b}(u,v)$, 
$a\not=b$ of the formal parameters $u$ and $v$ as follows
\begin{equation}\label{g-h-f-gen}
g(u,v)=\frac{c}{u-v},\quad 
h_{a}(u,v)=\fb_{a,a}+g(u,v)^{-1},\quad 
h_{a,b}(u,v)=2\,\fb_{a,b}+g(u,v)^{-1},\quad a\not=b\,.
\end{equation}

For any two labels $a,b\in\PSR_\ggo$, the
polynomial  functions   $\cre_{a,b}(u,v)$
are defined as follows
\begin{equation}\label{CuCu}
\cre_{a,b}(u,v)=\begin{cases}
 h_{a,b}(u,v),&a<b,\quad \fb_{a,b}\not=0\,,\\
  g(u,v)^{-1},&a> b,\quad \fb_{a,b}\not=0\,,\\
h_{a}(u,v),&a=b,\quad \fb_{a,b}\not=0\,,\\
1,& \scp_{a,b}=0\,.
\end{cases}
\end{equation}

\begin{Def}\label{aDYn}
The relations between generators of the asymmetric Yangian 
double  $\aDY{\fg}$ are given 
by the relations between currents ($a,b\in\PSR_\fg$)
\begin{subequations}\label{CuCoRe}
\begin{equation}\label{KK}
\aH^{q_1}_a(u)\,\aH^{q_2}_b(v)=\aH^{q_2}_b(v)\,\aH^{q_1}_a(u),\quad 
q_{1,2}=\pm\,,
\end{equation}
\begin{equation}\label{KF}
\cre_{a,b}(u,v)\,\aH^{\pm}_a(u)\, \aF_b(v) +
\cre_{b,a}(v,u)\,\aF_b(v)\, \aH^{\pm}_a(u)=0\,,
\end{equation}
\begin{equation}\label{FF}
\cre_{a,b}(u,v)\, \aF_a(u)\,\aF_b(v)+  \cre_{b,a}(v,u)\, \aF_b(v)\,\aF_a(u)=0\,,
\end{equation}
\begin{equation}\label{KE}
\cre_{a,b}(u,v)\, \aE_b(v)\, \aH^{\pm}_a(u)+
\cre_{b,a}(v,u)\, \aH^{\pm}_a(u)\, \aE_b(v)=0\,,
\end{equation}
\begin{equation}\label{EE}
\cre_{a,b}(u,v)\, \aE_b(v)\,\aE_a(u)+ \cre_{b,a}(v,u)\, \aE_a(u)\,\aE_b(v)=0\,,
\end{equation}
\begin{equation}\label{EF}
\begin{split}
[\aE_a(u),\aF_b(v)]
&=c\ \delta_{a,b}\ \delta(u,v)\Big(\aH^+_{a}(u)-\aH^-_{a}(v)\Big)\,,
\end{split}
\end{equation}
\begin{equation}\label{SR-gF}
\mathop{\rm Sym}\limits_{u_1,\ldots,u_{\fm_{a,b}}}
\Big[\aF_a(u_1),\ldots[\aF_a(u_{\fm_{a,b}}),\aF_{b}(v)]\ldots\Big]=0\,,
\end{equation}
\begin{equation}\label{SR-gE}
\mathop{\rm Sym}\limits_{u_1,\ldots,u_{\fm_{a,b}}}
\Big[\aE_a(u_1),\ldots[\aE_a(u_{\fm_{a,b}}),\aE_{b}(v)]\ldots\Big]=0\,.
\end{equation}
\end{subequations}
In the relation \r{EF} $\delta(u,v)$ is  the formal series 
\begin{equation}\label{delta}
\delta(u,v)=\frac{1}{u}\sum_{\ell\in\ZZ}\frac{v^\ell}{u^\ell}\,.
\end{equation}
In \r{SR-gF} and \r{SR-gE} $a\not=b$, $\fm_{a,b}=1-\ago_{a,b}$ and 
${\rm Sym}_{u_1,\ldots,u_{\fm_{a,b}}}$ stands for  
 the sum over all permutations of the formal 
 parameters $u_1,\ldots,u_{\fm_{a,b}}$:
\begin{equation}\label{Sym}
\mathop{\rm Sym}\limits_{u_1,\ldots,u_{\fm_{a,b}}}G(u_1,\ldots,u_{m_{a,b}})=
\sum_{\sigma\in S_{\fm_{a,b}}}G(u_{\sigma(1)},\ldots,u_{\sigma(\fm_{a,b})})
\end{equation}
for any formal series $G(u_1,\ldots,u_{\fm_{a,b}})$ of the parameters 
$u_1,\ldots,u_{\fm_{a,b}}$. The relations \r{SR-gF} and \r{SR-gE} are called the 
{\sl Serre relations}.
All the relations are understood as formal power series in $u$ and $v$.
\end{Def}

Let us define in the symmetric version of the 
Yangian double $\sDY{\fg}$ the currents $\yF_a(u)$, $\yE_a(u)$, and 
$\yH^\pm_a(u)$ by the formulas similar to \r{gen-ser-g}.

\begin{Def}\label{sDYn}
The commutation relations between generators of the symmetric Yangian 
double  $\sDY{\fg}$ are given by the same relations between currents
as in definition~\ref{aDYn} of the asymmetric Yangian double with replacement 
of the coefficient functions 
$\cre_{a,b}(u,v)\to  \cre^{\rm D}_{a,b}(u,v)$, where 
\begin{equation}\label{ocre}
\cre_{a,b}^{\rm D}(u,v)=\begin{cases}
 h^{\rm D}_{a,b}(u,v),& \fb_{a,b}\not=0\,,\\
1,& \scp_{a,b}=0\, ,
\end{cases}
\end{equation}
for any two labels $a,b\in\PSR_\ggo$ and
\begin{equation}
    h^{\rm D}_{a,b}(u,v) = \fb_{a,b}+g(u,v)^{-1}.
\end{equation}
\end{Def}

\begin{prop}\label{asDY}
The asymmetric  $\aDY{\fg}$ and symmetric $\sDY{\fg}$ versions of the Yangian 
doubles  are isomorphic as associative algebras. 
\end{prop} 

\proof
Let for $m\in\ZZ$
\begin{equation}\label{sign}
\sgn(m)=\begin{cases}
1,&m>0\,,\\ 0,&m=0\,,\\ -1,&m<0
\end{cases}
\end{equation}
be the sign function. Since, the functions $\cre_{a,b}(u,v)$ \r{CuCu} 
and $\cre^{\rm D}_{a,b}(u,v)$ \r{ocre} are related as follows
\begin{equation}\label{cre-creD}
    \cre_{a,b}(u,v) = \cre^{\rm D}_{a,b}
    (u + \sgn(b-a)\, c\, \fb_{a,b}\,  ,\,v)
\end{equation}
we can define a shift parameter $x_a$ for each node of the Dynkin diagram $a\in\PSR_\ggo$ such that
\begin{equation}
    x_b - x_a =  \sgn(b-a)\, \fb_{a,b} \quad \text{if}\quad \fb_{a,b}\not=0\,.
\end{equation}
Then for $\fg \ne \fe_{n}$ 
\begin{equation}\label{x-par}
    x_a = \sum_{b < a,\ b\in\PSR_{\fg}} \fb_{b,b+1}\,,\quad a\in\PSR_{\fg}
\end{equation}
and for $\fg = \fe_{n}$ 
\begin{equation*}
    x_a = 
    \begin{cases}
        \sum_{b < a} \fb_{b,b+1}, & a < n \\
       x_{n-3} +  \fb_{n-3,n}, & a = n.
    \end{cases}
\end{equation*}

Using the parameters $x_a$ we define the map $\tau$ acting on the currents 
of the asymmetric Yangian double $\aDY{\fg}$ by the
shifts of the formal parameters
\begin{equation}\label{tau-map}
\begin{aligned}
\tau:& &\aH^\pm_a(u)&\to \aH^\pm_a(u+c\,x_a) =\yH^\pm_a(u)\,,\\
\tau:& &\aF_a(u)&\to \aF_a(u+c\,x_a)\,\ =\yF_a(u)\,,\\
\tau:& &\aE_a(u)&\to \aE_a(u+c\,x_a)\ =\yE_a(u)\,.
\end{aligned}
\end{equation}
The map $\tau$ transforms the commutation relation \r{CuCoRe} of the asymmetric 
version of the Yangian double $\aDY{\fg}$ to the commutation relations 
of the definition~\ref{sDYn} for the symmetric version of the Yangian 
double $\sDY{\fg}$. 

This proves the proposition~\ref{asDY}. Since the map $\tau$ \r{tau-map} 
does not alter the inclusions  
$\aY{\fg}\hookrightarrow\aDY{\fg}$ and $\sY{\fg}\hookrightarrow\sDY{\fg}$
this also proves the proposition~\ref{asY}. 
\qed

In what follows we will consider only the asymmetric version of the Yangian 
double $\aDY{\fg}$ and will investigate only the algebraic 
properties of the currents $\aF_a(u)$ which follows from the 
commutation relations \r{FF} and the Serre relations \r{SR-gF}.

\subsection{Analytical properties of the currents}\label{cc-property}

The commutation relations \r{CuCoRe} 
in the asymmetric Yangian double demonstrates that there is a 
subalgebra $Y_\aF\subset\aDY{\fg}$ 
generated by the modes $\aF_a[\ell]$, $a\in\PSR_\fg$, $\ell\in\ZZ$ 
with the relations given by the commutation relation 
\r{FF} and the Serre relation \r{SR-gF}.

Although the currents are defined as formal series depending on formal parameters 
one can assign certain analytical properties 
to the products of the currents $\aF_a(u)\aF_b(v)$ 
which follow from their algebraic properties
when representations of the whole algebra $\aDY{\fg}$ are
restricted to the category of the highest weight representations.
In this case one can treat these products as analytical operator valued functions
in certain domains of the complex parameters $u$ and $v$.

We consider a completion of the subalgebra $Y_\aF\in\aDY{\fg}$ according to the  following 
definition.
\begin{Def}\label{Comp}
 The completed algebra $\overline{Y}_\aF$ is the vector space
of series given as infinite sums of monomials 
$\yF_{a_1}[\ell_1]\cdots \yF_{a_k}[\ell_k]$ with $\ell_1\leq \cdots\leq \ell_k$
such that $\ell_1+\cdots+\ell_k$ is fixed. The elements of the completed algebra
 $\overline{Y}_\aF$ are well-defined operators in any highest weight 
 representation of the Yangian double $\aDY{\fg}$.
\end{Def}

In order to compare any monomials composed from the modes of the 
currents $\aF_a[\ell]$, $a\in\PSR_\fg$, $\ell\in\ZZ$ one first has to normal order 
them according to the definition \r{Comp} using the commutation relations \r{cr-m5}
between modes. 

Let $\cV$  be a highest weight representation 
 of $\aDY{\fg}$ and $\cV^*$ be a dual highest weight representation. 
 Let $\ups$, $\ups^*$ be two generic vectors in $\cV$ and $\cV^*$ respectively.
Consider the matrix element of the product of the currents between these vectors
$\langle\ups^*,\aF_b(u)\,\aF_a(v)\,\ups\rangle$
for $a\leq b$, $\fb_{a,b} \ne 0$ such that this 
matrix element is not vanishing. 

Consider first the relation 
\r{FF} for $a=b$. By induction one can prove following proposition
\begin{prop}\label{no-aa}
In the Yangian double restricted to the category of the highest weight representation 
the relation \r{FF} for $a=b$ is equivalent to the relation 
\begin{equation}\label{no-aa1}
[\aF_a[\ell_1],\aF_a[\ell_2]]=\fb_{a,a}\sum_{\ell\geq0}
\Big(\aF_a[\ell_2-\ell-1]\,\aF_a[\ell_1+\ell]-\aF_a[\ell_1-\ell-1]\,\aF_a[\ell_2+\ell]\Big)
\end{equation}
for any fixed $\ell_1,\ell_2\in\ZZ$. 
\end{prop}

The relation \r{no-aa1} between modes of the simple root current $F_a(u)$ can be 
written as the commutation relation 
\begin{equation}\label{no-aa2}
f_a(u,v)\,\aF_a(u)\,\aF_a(v)=f_a(v,u)\,\aF(v)\,\aF(u)\,,
\end{equation}
where the function $f_a(u,v)=h_a(u,v)\,g(u,v)$ (resp. $f_a(v,u)$) in the left (resp. right) side of the equality 
\r{no-aa2} should be understood as a series over the ratio $v/u$ (resp. $u/v$). 

Then the commutation relation \r{FF} takes the form 
\begin{equation}\label{FFalt}
\ocre_{a,b}(u,v)\, \aF_a(u)\,\aF_b(v)=  \ocre_{b,a}(v,u)\, \aF_b(v)\,\aF_a(u)\,,
\end{equation}
where 
\begin{equation}\label{CuCu-alt}
\ocre_{a,b}(u,v)=\begin{cases}
 h_{a,b}(u,v)&a<b,\quad \fb_{a,b}\not=0\,,\\
  g(v,u)^{-1}&a> b,\quad \fb_{a,b}\not=0\,,\\
f_{a}(u,v),&a=b,\quad \fb_{a,b}\not=0\,,\\
1,& \scp_{a,b}=0\,.
\end{cases}
\end{equation}
Note that the sign in the case $a > b$ has been changed compared to eq. \eqref{CuCu}.

Let us show that the normal ordering 
of the current modes in the product  $\aF_b(u)\,\aF_a(v)$,  $a<b$, $\fb_{a,b} \ne 0$
provides poles of the function 
$\langle \ups^*,\aF_b(u)\,\aF_a(v)\,\ups\rangle\in\CC$ when $u\to v$. 
To do this we need the following proposition which can be proved by induction.
\begin{prop}\label{norm-ord}
The commutation relations \r{cr-m5} for $a<b$, $\ell_1<0$, $\ell_2\geq 0$ can be 
rewritten as the normal ordering relation 
\begin{equation}\label{no1}
\begin{split}
\aF_b[\ell_2]\,\aF_a[\ell_1]&=\aF_a[\ell_1]\,\aF_b[\ell_2]+\\
&\quad+\Theta(\ell_1+\ell_2)\,\Big[\aF_b[\ell_1+\ell_2],\aF_a[0]\Big]+
\Theta(-1-\ell_1-\ell_2)\,\Big[\aF_b[0],\aF_a[\ell_1+\ell_2]\Big]\\
&\qquad-2\,\fb_{a,b}\sum_{\ell\geq 0}\Theta(\ell_2-\ell-1)\Theta(-\ell_1-\ell-1)
\aF_a[\ell_1+\ell]\,\aF_b[\ell_2-\ell-1]\,,
\end{split}
\end{equation}
where $\Theta(\ell)$, $\ell\in\ZZ$ is the Heaviside step-function defined as follows
\begin{equation*}
\Theta(\ell)=\begin{cases}
1,&\ell\geq 0\,,\\ 0,&\ell<0\,.
\end{cases}
\end{equation*}
\end{prop}
Note that the r.h.s. of the equality \r{no1} always has a finite number of normal ordered terms 
for any fixed $\ell_1<0$ and $\ell_2\geq0$.

For any formal series $\yG(z)=\sum_{\ell\in\ZZ}\yG[\ell]z^{-\ell-1}$ we denote 
$\yG^{+}(z)=\sum_{\ell\geq 0}\yG[\ell]z^{-\ell-1}$, and 
$\yG^{-}(z)=-\ \sum_{\ell< 0}\yG[\ell]z^{-\ell-1}$. The simple root current 
$\aF_a(z)$, $a\in\PSR_\fg$ can be presented as a linear combination 
$\aF_a(z)=\aF^{+}_a(z)-\aF^{-}_a(z)$ and the product of two 
simple root currents $\aF_b(u)\,\aF_a(v)$ as the sum of four terms 
\begin{equation}\label{hwp6}
\begin{split}
\aF_b(u)\,\aF_a(v)&=\aF^{+}_b(u)\,\aF^{+}_a(v)+\aF^{-}_b(u)\,\aF^{-}_a(v)-\\
&\qquad -\ \aF^{-}_b(u)\,\aF^{+}_a(v)-\aF^{+}_b(u)\,\aF^{-}_a(v)\,.
\end{split}
\end{equation}

It is clear that the first three terms in the r.h.s. of \r{hwp6} belongs to 
completed subalgebra $\oY_\aF$, while the last one does not. 
In order to consider the product of the currents $\aF_a(z)\,\aF_b(w)$
as an element of $\oY_\aF$ we have to exchange the modes
 $\aF_b[\ell_2]\,\aF_a[\ell_1]$ with $\ell_2\geq0>\ell_1$ in the product of the half-currents 
$\aF^{+}_b(u)\,\aF^{-}_a(v)$ using the commutation 
relations \r{cr-m5} in the form \r{no1}. 
After this normal ordering procedure 
the product $\aF_b(u)\,\aF_a(v)$ will belong to the completed subalgebra 
$\oY_\aF$ and the analytical structure of the function 
$\langle\ups^*,\aF_b(u)\,\aF_a(v)\,\ups\rangle$ in the domain 
$|u|\gg|v|$ is revealed
\begin{equation}\label{hwp7}
\begin{split}
&\langle\ups^*,\aF_b(u)\,\aF_a(v)\,\ups\rangle=\\
&\quad=
\langle\ups^*,(\aF^{+}_b(u)\,\aF^{+}_a(v)+\aF^{-}_b(u)\,\aF^{-}_a(v) - \aF^{-}_b(u)\,\aF^{+}_a(v)
-\aF^{-}_a(v)\,\aF^{+}_b(u))\,\ups\rangle\\
&\qquad+
\frac{c}{u-v}\,\langle\ups^*,(2\,\fb_{a,b}\,\aF^{-}_a(v)\,\aF^{+}_b(u)
+[\aF^+_b(u),\aF_a[0]]-[\aF_b[0],\aF^-_a(v)])\,\ups\rangle\,.
\end{split}
\end{equation}
Here the analytical function $1/(u-v)$ is obtained as a series $u^{-1}\sum_{\ell\geq 0}
(v/u)^\ell$ in the domain $|u|\gg|v|$. Loosely speaking and due to the commutation relation 
\r{cr-m5} we may consider the product of the currents $\aF_b(u)\,\aF_a(v)$ 
as an operator valued analytical function in the domain  $|u|\gg|v|$ with a simple 
pole at $u=v$. 

Analogously one can show that the product $\aF_a(v)\,\aF_b(u)$ may be treated 
as the analytical operator valued function in the domain $|u|\ll|v|$ with 
a simple pole at $v=u+2\,c\,\fb_{a,b}$. 
Then the analytical continuation 
of this product to the domain $|u|\gg|v|$ will coincide with 
the element of the completed subalgebra $\oY_\aF$
\begin{equation*}
\frac{\ocre_{b,a}(u,v)}{\ocre_{a,b}(v,u)}\ \aF_b(u)\,\aF_a(v)
\end{equation*}
which should be also normal ordered. This demonstrates that in order to 
compare any products of the currents in the category of the highest weight 
representations we first have to make 
sure that all monomials in this product are normal ordered or 
lie in the completed subalgebra $\oY_\aF$. 
We will detailed these normal ordering calculations in the appendix~\ref{AppA}
for the small rank algebra $\fg=\fg\fl_3$. 

Following the ideas of the papers \cite{E,DKhP} one can extend 
the arguments above to the 
 matrix coefficient 
 \begin{equation}\label{hwp1}
 \begin{split}
 &\langle\ups^*,\aF_{a_1}(z_1)\cdots \aF_{a_k}(z_k)\,\ups\rangle=\\
&\qquad= \sum_{\ell_1,\ldots,\ell_k\in\ZZ}
 \langle\ups^*,\aF_{a_1}[\ell_1]\cdots \aF_{a_k}[\ell_k]\,\ups\rangle\ 
 (z_1/c)^{-\ell_1-1}\cdots (z_k/c)^{-\ell_k-1}
 \end{split}
 \end{equation}
 as a formal power series over complex parameters $z_1,\ldots,z_k$
 such that not all coefficients 
 $\langle\ups^*,\aF_{a_1}[\ell_1]\cdots \aF_{a_k}[\ell_k]\,\ups\rangle$
 are vanishing. 
 Using the commutation relations between modes of the currents 
 $\aF_a(z)$ \r{cr-m5} in the form of the normal ordering relations \r{no1}
 allows to demonstrate that  the formal power series \r{hwp1} 
 \begin{equation}\label{hwp2}
 \langle\ups^*,\aF_{a_1}(z_1)\cdots \aF_{a_k}(z_k)\,\ups\rangle\in
 \CC\Big[\frac{z_1}{c},\frac{c}{z_1},\ldots,\frac{z_k}{c},\frac{c}{z_k}\Big]\ \Big[\Big[
 \frac{z_2}{z_1}, \frac{z_3}{z_2},\ldots, \frac{z_{k-1}}{z_{k-2}}, \frac{z_{k}}{z_{k-1}}
 \Big]\Big]
 \end{equation}
 can be presented as Taylor series over the variables 
 $ z_2/z_1, z_3/z_2,\ldots, z_{k-1}/z_{k-2}, 
 z_{k}/z_{k-1}$ with coefficients being polynomials 
 over ${z_1}/{c},{c}/{z_1},\ldots,{z_k}/{c},{c}/{z_k}$. Repeating the
 arguments of the paper \cite{E} this signifies that the formal 
 power series \r{hwp1} converges to a meromorphic function. 
 Then the commutation relations \r{FFalt} allows to  characterize partially 
 this meromorphic function. 
 
 Let $a_i\in\PSR_\fg$ for $i=1,\ldots,k$.
 Let us consider the matrix coefficient between fixed degree and dual 
  fixed degree vectors $\ups$ and $\ups^*$  of the relation 
 \begin{equation}\label{hwp3} 
 \prod_{i<j}^k\ocre_{a_i,a_j}(z_i,z_j)\,\aF_{a_1}(z_1)\cdots \aF_{a_k}(z_k)=
 \prod_{i<j}^k\ocre_{a_j,a_i}(z_j,z_i)\,\aF_{a_k}(z_k)\cdots \aF_{a_1}(z_1)\,.
 \end{equation}
 The formal series equality \r{hwp3} is a direct consequence of the commutation 
 relations \r{FFalt}. The arguments of the paper \cite{E} which were shown 
 above yield  
  \begin{equation}\label{hwp4}
 \begin{split}
  &\prod_{i<j}^k\ocre_{a_i,a_j}(z_i,z_j)\,
  \langle\ups^*,\aF_{a_1}(z_1)\cdots \aF_{a_k}(z_k)\,\ups\rangle\in\\
&\qquad\in 
 \CC\Big[\frac{z_1}{c},\frac{c}{z_1},\ldots,\frac{z_k}{c},\frac{c}{z_k}\Big]\ \Big[\Big[
 \frac{z_2}{z_1}, \frac{z_3}{z_2},\ldots, \frac{z_{k-1}}{z_{k-2}}, \frac{z_{k}}{z_{k-1}}
  \Big]\Big]\,.
 \end{split}
 \end{equation}
 Repeating the same arguments to the r.h.s. of the relation \r{hwp4} 
 one gets 
  \begin{equation}\label{hwp5}
 \begin{split}
  &\prod_{i<j}^k\ocre_{a_j,a_i}(z_j,z_i)\,
  \langle\ups^*,\aF_{a_k}(z_k)\cdots \aF_{a_1}(z_1)\,\ups\rangle\in\\
&\qquad\in 
 \CC\Big[\frac{z_1}{c},\frac{c}{z_1},\ldots,\frac{z_k}{c},\frac{c}{z_k}\Big]\ \Big[\Big[
 \frac{z_{k-1}}{z_k}, \frac{z_{k-2}}{z_{k-1}},\ldots, \frac{z_{2}}{z_{3}},
 \frac{z_{1}}{z_{2}}
\Big]\Big]\,.
 \end{split}
 \end{equation}
 The equality \r{hwp3} means that the matrix coefficients of the l.h.s. and the r.h.s.
 of this equality belong to the intersection of the spaces 
 \begin{equation*}
  \CC\Big[\frac{z_1}{c},\frac{c}{z_1},\ldots,\frac{z_k}{c},\frac{c}{z_k}\Big]\ \Big[\Big[
 \frac{z_2}{z_1}, \ldots, \frac{z_{k-1}}{z_{k-2}}, \frac{z_{k}}{z_{k-1}}
 \Big]\Big]\cap
 \CC\Big[\frac{z_1}{c},\frac{c}{z_1},\ldots,\frac{z_k}{c},\frac{c}{z_k}\Big]\ \Big[\Big[
 \frac{z_{k-1}}{z_k}, \ldots, \frac{z_{2}}{z_{3}}, \frac{z_{1}}{z_{2}}
 \Big]\Big]
 \end{equation*}
 namely, to the space of polynomials 
 \begin{equation*}
  \CC\Big[\frac{z_1}{c},\frac{c}{z_1},\ldots,\frac{z_k}{c},\frac{c}{z_k}\Big]
 \end{equation*}
 over variables ${z_1}/{c},{c}/{z_1},\ldots,{z_k}/{c},{c}/{z_k}$. It means 
 that the series \r{hwp1} converges in the domain $z_j\gg z_i$ for 
 $i<j$  to a meromorphic function in $(\CC^*)^k$  
  with simple poles (zeros) defined by the simple zeros (poles) of the function 
 $\prod_{i<j}\ocre_{a_i,a_j}(z_i,z_j)$. 
 
 Hence, in the category of the highest weight representations, one can 
 treat the generating functions $\aF_{a_1}(z_1)\cdots \aF_{a_k}(z_k)$
 as operator valued functions, analytical in the region 
 $|z_1|\gg|z_2|\gg\cdots\gg|z_k|$. Due to the fact that 
 the product of Taylor series is well defined, the product 
 $\aF_{a_1}(z_1)\cdots \aF_{a_k}(z_k)$
  coincides with the  
 product $\Big(\aF_{a_1}(z_1)\cdots \aF_{a_\ell}(z_\ell)\Big)
 \cdot\Big(\aF_{a_{\ell+1}}(z_{\ell+1})\cdots \aF_{a_k}(z_k)\Big)$ for 
 any $\ell=1,\ldots,k-1$. Moreover, the commutation relations \r{FFalt}
 describe the analytical continuation of the operator valued 
 functions to other regions
 explained below using the normal ordering procedure.

Once again, we have argued above, following the paper \cite{E}, that the 
restriction of the Yangian double $\aDY{\fg}$ realized in terms of the currents 
\r{gen-ser-g}  to the category of the highest weight representations
leads to certain analytical properties of the products of  currents 
$ \aF_{a_1}(z_1)\cdots \aF_{a_k}(z_k)$
dictated by the commutation relations \r{FFalt} and 
the Serre relations \r{SR-gF}. We will explore these 
analytical properties below to describe the composed currents as elements 
of the completed subalgebra $\oY_\aF$. 
We will also  demonstrate in this case that 
the Serre relations for the simple root currents  can be reduced to the quadratic commutation relations for the composed currents. 
We will describe the composed currents associated with 
all positive roots of the algebra $\fg$
only for the classical series and
postpone such description for exceptional algebras to our future 
publications. Nevertheless, we will demonstrate certain 
analytical properties of the particular composed current in the Yangian 
double $\aDY{\fg_2}$ for the exceptional algebra $\fg=\fg_2$ in proposition 
\ref{Enr5} and appendix~\ref{AppB}.

The analytical properties of the products of  currents $\aE_a(u)$ can be 
similarly considered.  We do not describe these properties in this paper.

\section{Analytical properties and Serre relations}
\label{an-pro-BV}

The goal of this section is to demonstrate that the Serre relations for 
the currents $\aF_a(z)$, $a\in\PSR_\fg$ lead to the apparition of additional zeros 
in the products of these currents. This result is similar to the ones obtained in \cite{E} for the quantum affine algebras.

When $\fm_{a,b}=1$ 
($\fa_{a,b}=0$)
the Serre relations are equivalent to the commutativity of the 
currents $[\aF_a(v),\aF_b(u)]=0$ \r{FFalt}.
When $\ago_{a,b}\not =0$, 
the positive integer $\fm_{a,b}=2$ ($\fa_{a,b}=-1$) for all simply laced algebras 
$\fs\fl_n$, $\fo_{2n}$, $\fe_6$, $\fe_7$, $\fe_8$
and the Serre relations \r{SR-gF} become 
 third order relations between currents. 
It is also  equal to 2 in most of the cases for other simple Lie algebras 
except for the cases $\fm_{a,b}=3$ ($\fa_{a,b}=-2$) and 
 $\fm_{a,b}=4$ ($\fa_{a,b}=-3$). The former case 
  happens when $a=0$, $b=1$ for  $\fg=\ggb$; 
 $a=2$, $b=3$ for  $\fg=\ff_4$,
  and  $a=1$, $b=0$ for $\fg=\gsp$. 
 The latter case arises when $a=2$, $b=1$ for $\fg=\fg_2$.
 In these cases the Serre relations become 
forth and fifth order relations respectively.

\subsection{Serre strata}\label{sect31}

For any two different nodes $a\neq b$ of the Dynkin diagram of $\ggo$ 
such that $\fb_{a,b}\neq 0$ we consider the cardinality 1 set $\bt^b=\{t^b\}$
and the cardinality $\fm_{a,b}=1-\ago_{a,b}$ set 
$\bt^a=\{t^a_1,\ldots,t^a_{\fm_{a,b}}\}$.   
We use these notations for the parameters in view of further applications 
to $\fg$-invariant integrable models when the parameters 
$t^s_\ell$ become the Bethe parameters of the corresponding off-shell 
Bethe vectors \cite{LPR-BV-CU}.

Let us consider the following  element from the subalgebra $\oY_\aF$
\begin{equation}\label{Ser-ele}
\begin{split}
\sF_{a,b}(\bt^a,t^b)&=
\prod_{\ell_1<\ell_2}^{\fm_{a,b}}\ocre_{a,a}(t^a_{\ell_1},t^a_{\ell_2})\ 
\prod_{\ell=1}^{\fm_{a,b}}
\ocre_{b,a}(t^b,t^a_\ell)\ \aF_b(t^b)\,
\aF_a(t^a_1)\cdots \aF_a(t^a_{\fm_{a,b}})=\\
 &=
\prod_{\ell_1<\ell_2}^{\fm_{a,b}}\ocre_{a,a}(t^a_{\ell_1},t^a_{\ell_2})\ 
\prod_{\ell=1}^{\fm_{a,b}}
\ocre_{a,b}(t^a_\ell,t^b)\ 
\aF_a(t^a_1)\cdots \aF_a(t^a_{\fm_{a,b}})\, \aF_b(t^b)\,.
 \end{split}
\end{equation}
The set $\bt^a$ in \eqref{Ser-ele} has  
cardinality $\fm_{a,b}=2,3,4$.
The equality between the first and the second line in \r{Ser-ele} as well as the
symmetry  of this element  with respect to permutations 
in the set $\bt^a$ follow from the commutation relations 
\r{FFalt}. 

There are two approaches to prove the vanishing of the elements
 \r{Ser-ele} on some hyperplan (called Serre strata) in the $\{\bar t^a,t^b\}$ space. 
 One approach is due 
to B.~Enriquez \cite{E}. It uses certain identities for the distributions
associated to the  Serre relations  
 to prove the vanishing properties 
of the elements \r{Ser-ele}. We  sketch this approach in the 
appendix~\ref{AppC}.  

The second approach was introduced in the paper \cite{DKh} by 
considering the composed currents. It was shown  
that the commutation relations 
of the composed currents with the simple root currents are equivalent 
to the Serre relations. 
We follow this second strategy in the following propositions.

To make an analogy with simple Lie algebras, one can view the definitions \eqref{aYn} or \eqref{sYn} as Serre-Chevalley bases of the Yangian, while the use of composed current leads to a Cartan-Weyl basis for the same Yangian.

We consider the Serre relation
\begin{equation}\label{Ser1}
\sum_{\sigma\in S_2}
\Big[\aF_a(t^a_{\sigma(1)}),\Big[\aF_a(t^a_{\sigma(2)}),\aF_b(t^b)\Big]\Big]=0\,.
\end{equation}
\begin{prop}\label{Enr1}
 Let $\fb_{a,b}\not=0$, $\fm_{a,b}=2$ and $a<b$. We consider the element  
\begin{equation}\label{pr-cu-1}
\sF_{a,b}(\bt^a,t^b)=
\frac{f_a(t^a_2,t^a_1)}{g(t^a_2,t^b)\,g(t^a_1,t^b)}\,
\aF_{b}(t^b)\, \aF_a(t^a_2)\, \aF_a(t^a_1)
\end{equation} 
(i) The element \r{pr-cu-1}, is well-defined 
considered as an operator valued function
of the complex parameters 
$t^b-t^a_1$, $t^b-t^a_2$, $t^a_1-t^a_2$,  with no poles and sigularities only at infinity.
\\
(ii) This operator valued function 
has a simple zero 
\begin{equation}\label{zero1}
\sF_{a,b}(\{t^b,t^b-2\,c\,\fb_{a,b}+c\,\varepsilon\},t^b)\Big|_{\varepsilon\to 0}=
2\,\varepsilon\,\fb_{a,b}\ \aF_{a}(t^b)\,
\aF_{b}(t^b)\,\aF_{a}(t^b-2\,c\,\fb_{a,b})
\end{equation}
on each Serre stratum
\begin{equation}\label{S-st1}
t^b=t^a_{\sigma(1)}=t^a_{\sigma(2)}+2\,c\,\fb_{a,b}\,,
\end{equation} 
where $\sigma\in S_2$.
The product $\aF_{a}(t^b)\,
\aF_{b}(t^b)\,\aF_{a}(t^b-2\,c\,\fb_{a,b})$ is a well defined element 
in the completed algebra $\oY_\aF$.
\\
(iii) 
The same element \r{pr-cu-1} is regular on the each stratum 
\begin{equation}\label{S-st1a}
t^b=t^a_{\sigma(1)}=t^a_{\sigma(2)}-2\,c\,\hgo_{a,b}
\end{equation}
 where it is equal to
\begin{equation}\label{Areg1}
 \sF_{a,b}(\{t^b,t^b+2\,c\,\fb_{a,b}\},t^b)= 8\,\fb^2_{a,b}\,
 \aF_{a}(t^b)\,\aF_{b}(t^b)\,\aF_{a}(t^b+2\,c\,\fb_{a,b})\,.
 \end{equation}
\end{prop}

\proof 
Rewrite the element \r{pr-cu-1} using the 
commutation relation \r{FF} in the following form 
 \begin{equation}\label{Ep1}
 \sF_{a,b}(\bt^a,t^b)=\frac{f_a(t^a_2,t^a_1)\,h_{a,b}(t^a_2,t^b)}
 {g(t^a_1,t^b)}\,
 \aF_a(t^a_2)\, \aF_{b}(t^b)\, \aF_a(t^a_1)\,.
 \end{equation}
Close to the Serre stratum \r{S-st1} we can set $t^a_2=t^b$ 
in \r{Ep1}  to get  
\begin{equation}\label{Ep2}
 \sF_{a,b}(\{t^a_1,t^b\},t^b)=-\ 2\,\fb_{a,b}\,h_a(t^b,t^a_1)\,
\aF_a(t^b)\, \aF_{b}(t^b)\,  \aF_a(t^a_1)\,.
\end{equation}

The fact that the product of the composed current 
$\aF_a(t^b)\, \aF_{b}(t^b)$ with the simple root current $\aF_a(t^a_1)$ 
is non-singular (i.e. has no zeros nor singularities for all parameters $t^b$ and $t^a_1$)
in the asymmetric Yangian double $\aDY{\fg}$ restricted to the category 
of the highest weight representations 
follows from the commutation relation 
\begin{equation}\label{Ep4}
\big(\aF_a(t^b)\, \aF_{b}(t^b)\big)\, \aF_a(t^a_1)=f_a(t^a_1,t^b)\, \aF_a(t^a_1)\, 
\big(\aF_a(t^b)\, \aF_{b}(t^b)\big)
\end{equation}
which is implied by the corresponding Serre relation \r{Ser1}
(see the commutation relation \r{Ac92} below). 
We will prove this implication
in the section~\ref{CC-SR}. 
In the category of the highest weight representations the commutation relation 
\r{Ep4} signifies that the product of the currents 
$\big(\aF_a(t^b)\, \aF_{b}(t^b)\big)\, \aF_a(t^a_1)$
in this order is a 'regular' function 
of $t^b$ and $t^a_1$ in the sense explained in the section~\ref{cc-property} 
and can be evaluated at any finite values of these parameters. 

In particular, one can consider element \r{Ep2} at the value
$t^a_1=t^b-2\,c\,\fb_{a,b}+c\,\epsilon$
and  the vanishing property \r{zero1} will follow from the equality 
\begin{equation}\label{Ep3}
\fb_{a,a}+2\,\fb_{a,b}=0\quad\mbox{for}\quad a<b,\quad 
 m_{a,b}=2\,.
\end{equation}
On the other hand, the same element  \r{Ep2} at the value 
$t^a_1=t^b+2\,c\,\fb_{a,b}$ is equal to \r{Areg1}. 
  \qed

The statements which describe the apparition  
of the zeros in other products of simple root currents of the type \r{Ser-ele}
can be similarly formulated. We do it below  in the next four propositions, sketching the proof of each proposition, but leaving the complete proof to the interested reader.

\begin{prop}\label{Enr2}
 Let $\fb_{a,b}\not=0$, $\fm_{a,b}=2$ and $a>b$. 
 \\
(i) The  element  
\begin{equation}\label{pr-cu-2}
\sF_{a,b}(\bt^a,t^b)=
\frac{f_a(t^a_2,t^a_1)}{g(t^b,t^a_2)\,g(t^b,t^a_1)}\,
\aF_a(t^a_2)\, \aF_a(t^a_1)\,\aF_{b}(t^b)
\end{equation} 
considered as an operator valued function
of the complex parameters 
$t^b-t^a_1$, $t^b-t^a_2$, $t^a_1-t^a_2$  is well-defined  
 with no poles and sigularities only at infinity.
\\
(ii) This operator valued function 
has a simple zero 
\begin{equation}\label{zero2}
\sF_{a,b}(\{t^b,t^b+2\,c\,\fb_{a,b}+c\,\varepsilon\},t^b)\Big|_{\varepsilon\to 0}=
-\ 2\,\varepsilon\,\fb_{a,b}\ \aF_{a}(t^b+2\,c\,\fb_{a,b})\,
\aF_{b}(t^b)\,\aF_{a}(t^b)
\end{equation}
on the each Serre stratum 
\begin{equation}\label{S-st2}
t^b=t^a_{\sigma(1)}=t^a_{\sigma(2)}-2\,c\,\fb_{a,b}\,,
\end{equation} 
where $\sigma\in S_2$.
The product of the currents $\aF_{a}(t^b+2\,c\,\fb_{a,b})\,
\aF_{b}(t^b)\,\aF_{a}(t^b)$ is a well defined element 
in the completed algebra $\oY_\aF$.
\\
(iii) The element \r{pr-cu-2}
is regular on  each stratum 
\begin{equation}\label{S-st2a}
t^b=t^a_{\sigma(1)}=t^a_{\sigma(2)}+2\,c\,\fb_{a,b}
\end{equation}
 where it is equal to
\begin{equation}\label{Areg2}
 \sF_{a,b}(\{t^b,t^b-2\,c\,\fb_{a,b}\},t^b)= -\ 8\,\fb^2_{a,b}\,
 \aF_{a}(t^b-2\,c\,\fb_{a,b})\,\aF_{b}(t^b)\,\aF_{a}(t^b)\,.
 \end{equation}
\end{prop}
\proof The proof is absolutely the same as the one of the proposition~\ref{Enr1}.
The only difference is that instead of the commutation relation \r{Ep4} 
one should use the commutation relation \r{Ac9}
\begin{equation*}
\aF_a(t^a_2)\,\big(\aF_b(t^b)\, \aF_{a}(t^b)\big)\, =f_a(t^b,t^a_2)\, 
\big(\aF_b(t^b)\, \aF_{a}(t^b)\big)\,\aF_a(t^a_2)\,. 
\end{equation*}
 \qed

The two next propositions are built on the Serre relation
\begin{equation}\label{Ser3}
\sum_{\sigma}
\Big[\aF_a(t^a_{\sigma(1)}),\Big[\aF_a(t^a_{\sigma(2)}),\Big[\aF_a(t^a_{\sigma(3)}),
\aF_b(t^b)\Big]\Big]\Big]=0.
\end{equation}
\begin{prop}\label{Enr3}
Let  $m_{a,b}=3$ and $a=0$, $b=1$,  $\ggo=\ggb$  or 
$a=2$, $b=3$,  $\ggo=\ff_4$ (in both cases $\fb_{b,b}=1$ and 
$\fb_{a,a}=-\fb_{a,b}=1/2$).
\\
(i) The element  
\begin{equation}\label{pr-cu-13}
\sF_{a,b}(\bt^a,t^b)=
\frac{f_a(t^a_2,t^a_1)\,f_a(t^a_3,t^a_1)f_a(t^a_3,t^a_2)}
{g(t^b,t^a_1)\,g(t^b,t^a_2)\,g(t^b,t^a_3)}\,
\aF_{b}(t^b)\, \aF_a(t^a_3)\, \aF_a(t^a_2)\, \aF_a(t^a_1)
\end{equation} 
considered as an operator valued  function
of the parameters  $t^b-t^a_s$, $s=1,2,3$
and $t^a_{s_1}-t^a_{s_2}$, 
$s_1,s_2=1,2,3$ is well defined with no poles and singularities only at infinity.
\\
(ii) This operator valued function 
has a simple zero 
\begin{equation}\label{zero13}
\begin{split}
&\sF_{a,b}(\{t^b,t^b-c\,\fb_{a,b},t^b-2\,c\,\fb_{a,b}+c\,\varepsilon\},t^b)
\Big|_{\varepsilon\to 0}=\\
&\qquad= 
-\ 2\,\varepsilon\,\fb_{a,b}\ \aF_{a}(t^b-c\,\fb_{a,b})\, 
\aF_{a}(t^b)\,\aF_{b}(t^b)\,\aF_{a}(t^b-2\,c\,\fb_{a,b})
\end{split}
\end{equation}
on  each Serre stratum
\begin{equation}\label{S-st13}
t^b=t^a_{\sigma(1)}=t^a_{\sigma(2)}+c\,\fb_{a,b}=
t^a_{\sigma(3)}+2\,c\,\fb_{a,b}\,,
\end{equation} 
where $\sigma$ is any permutation of the set $(1,2,3)$. 
The product of the currents 
\begin{equation*}\aF_{a}(t^b-c\,\fb_{a,b})\, 
\aF_{a}(t^b)\,\aF_{b}(t^b)\,\aF_{a}(t^b-2\,c\,\fb_{a,b})
\end{equation*}
is a non-vanishing 
well defined element 
in the completed algebra $\oY_\aF$.
\end{prop}
\proof
The proof is again similar to the one of the proposition 
\ref{Enr1}, using now the Serre relation \eqref{Ser3} and the commutation relation \r{FFalt}. Note that in 
\r{zero13}  the product 
$F_{a}(t^b-c\,\fb_{a,b})\, F_{a}(t^b)\,F_{b}(t^b)$ is an
example of a composed current for the Yangian doubles $\aDY{\ggb}$
and $\aDY{\ff_4}$. The fact that the product of the currents 
in the r.h.s. of \r{zero13} is well defined follows from the commutation relation 
\begin{equation*}
\begin{split}
&\Big(F_{a}(t^b-c\,\fb_{a,b})\, F_{a}(t^b)\,F_{b}(t^b)\Big)\,F_a(t^a_1)=\\
&\qquad = f(t^a_1,t^b+c/2)\, F_a(t^a_1)\,
\Big(F_{a}(t^b-c\,\fb_{a,b})\, F_{a}(t^b)\,F_{b}(t^b)\Big)
\end{split}
\end{equation*}
between this composed current and the simple root current $\aF_a(t^a_1)$. 
This is the first commutation relation in \r{B35} and its  implication by
 the Serre relation \r{Ser3} will be proved in the section~\ref{AppG} 
for the Yangian double $\aDY{\ggb}$.\qed

\begin{prop}\label{Enr4}
Let $m_{a,b}=3$, $a=1$, $b=0$, and $\fg=\gsp$  (in this case 
$\fb_{b,b}=2$ and $\fb_{a,a}=-\fb_{a,b}=1$). 
\\
(i) The element  
\begin{equation}\label{pr-cu-14}
\sF_{a,b}(\bt^a,t^b)=
\frac{f_a(t^a_2,t^a_1)\,f_a(t^a_3,t^a_1)f_a(t^a_3,t^a_2)}
{g(t^b,t^a_1)\,g(t^b,t^a_2)\,g(t^b,t^a_3)}\,
 \aF_a(t^a_3)\, \aF_a(t^a_2)\, \aF_a(t^a_1)\,\aF_{b}(t^b)
\end{equation} 
considered as an operator valued  function
of the parameters  $t^b-t^a_s$, $s=1,2,3$
and $t^a_{s_1}-t^a_{s_2}$, 
$s_1,s_2=1,2,3$ is well defined with no poles and singularities only at infinity.
\\
(ii) This operator valued function 
has a simple zero 
\begin{equation}\label{zero14}
\begin{split}
&\sF_{a,b}(\{t^b,t^b+c\,\fb_{a,b}+c\,\varepsilon,t^b+2\,c\,\fb_{a,b}\},t^b)
\Big|_{\varepsilon\to 0}=\\
&\qquad=
-\ 4\,\varepsilon\,\fb_{a,b}\ \aF_{a}(t^b+c\,\fb_{a,b})\, 
\aF_{a}(t^b+2\,c\,\fb_{a,b})\,\aF_{b}(t^b)\,\aF_{a}(t^b)
\end{split}
\end{equation}
on  each Serre stratum
\begin{equation}\label{S-st14}
t^b=t^a_{\sigma(1)}=t^a_{\sigma(2)}-c\,\fb_{a,b}=t^a_{\sigma(3)}-2\,c\,\fb_{a,b}\,,
\end{equation} 
where $\sigma$ is any parmutation of the set $(1,2,3)$. 
\end{prop}
\proof Once more, it is similar to the proof of proposition~\ref{Enr1}, using now the Serre relation \eqref{Ser3} and the commutation relation \r{FF}. Note that 
the product $F_1(t^b-2\,c)\,F_0(t^b)\,F_1(t^b)$ 
in the r.h.s. \r{zero14} is an example of composed 
current for the algebra $\aDY{\gsp}$. The fact that the product 
of the simple root current $\aF_{1}(t^b-c)$ and the composed 
root current $F_1(t^b-2\,c)\,F_0(t^b)\,F_1(t^b)$ follows from the 
commutation relation \r{C51} which in turn is implied by the Serre 
relation \r{Ser3} for the Yangian double $\aDY{\gsp}$ and $a=1$, $b=0$. \qed

Finally, we formulate an analogous proposition for the algebra $\fg=\fg_2$
and its Serre relation
\begin{equation}\label{Ser4}
\sum_{\sigma}
\Big[\aF_2(v_{\sigma(1)}),\Big[\aF_2(v_{\sigma(2)}),\Big[\aF_2(v_{\sigma(3)}),
\Big[\aF_2(v_{\sigma(4)}),\aF_1(u)\Big]\Big]\Big]\Big]=0.
\end{equation}
To lighten the presentation we 
rename $t^b=t^1\to u$ and $t^2_\ell=t^2_\ell\to v_\ell$, $\ell=1,2,3,4$.

\begin{prop}\label{Enr5}
Let $\fm_{a,b}=4$, $a=2$, $b=1$, and $\fg=\fg_2$. 
\\
(i) The element 
\begin{equation}\label{SE-G2}
\sF_{2,1}(u,\{v_1,v_2,v_3,v_4\})=\prod_{1\leq\ell_1<\ell_2\leq4}
\frac{v_{\ell_1}-v_{\ell_2}+c}{v_{\ell_1}-v_{\ell_2}}\,
\frac{\,\aF_2(v_1)\,\aF_2(v_2)\,\aF_2(v_3)\,\aF_2(v_4)\,\aF_1(u)}
{g(u,v_1)\,g(u,v_2)\,g(u,v_3)\,g(u,v_4)}
\end{equation}
which belongs to the completed subalgebra $\oY_\aF$ in the 
Yangian double $\aDY{\fg_2}$  and considered as an operator valued 
function of the complex parameters 
 $u-v_s$ and 
$v_{s_1}-v_{s_2}$, $s_1,s_2=1,\ldots,4$ is well defined with no poles and singularities 
only at infinity.
\\
(ii) 
This operator valued function has a simple zero
\begin{equation}\label{zero-G2}
\begin{split}
&\sF_{2,1}(u,\{u,u-c+c\,\epsilon,u-2\,c,u-3\,c\})\Big|_{\epsilon\to0}=\\
&\qquad=
24\,\epsilon\,\aF_2(u-2\,c)\,\aF_2(u-3\,c)\,\aF_1(u)\,\aF_2(u)\,\aF_2(u-c)
\end{split}
\end{equation}
 at  each Serre stratum 
\begin{equation}\label{G2-str}
u=v_{\sigma(1)}=v_{\sigma(2)}+c=v_{\sigma(3)}+2\,c=v_{\sigma(4)}+3\,c\,,
\end{equation}
where $\sigma$ is any permutation of the 
set $(1,2,3,4)$.
The product of the currents in the r.h.s. of \r{zero-G2} is a well defined 
element in the completed subalgebra $\oY_\aF$.
\end{prop}
\proof The techniques to prove this proposition is the same as in 
proposition~\ref{Enr1}, using the Serre relation \eqref{Ser4} and the commutation relations \eqref{FFalt}. 
It is based on the commutation relation 
between simple root current $\aF_2(v)$ and composed root current
$\aF_2(u-2\,c)\,\aF_2(u-3\,c)\,\aF_1(u)\,\aF_2(u)$:
\begin{equation}\label{sm-cr-g2}
\begin{split}
&\frac{u-v+c}{u-v}\,\Big(\aF_2(u-2\,c)\,\aF_2(u-3\,c)\,\aF_1(u)\,\aF_2(u)\Big)\,\aF_2(v) =\\
&\qquad=
\frac{v-u+4\,c}{v-u+2\,c}\,
\aF_2(v) \Big(\aF_2(u-2\,c)\,\aF_2(u-3\,c)\,\aF_1(u)\,\aF_2(u)\Big)\ 
\end{split}
\end{equation}
which is implied by  the Serre relation \r{Ser4}. This implication will be 
shown in appendix~\ref{AppB} using technique of the paper \cite{DKh}. \qed

\section{Composed currents and Serre relations}
\label{CC-SR}

The composed currents were introduced in \cite{DKh}, 
for the  quantum affine algebras 
$U_q(\widehat{\ggo})$, where $\ggo$ is a simply laced  Lie algebra. 
In particular, it was shown that the commutation relations 
of these composed currents with simple root currents are equivalent 
to the Serre relations between simple root currents. The main ingredient 
of the construction in \cite{DKh} was a Weyl group extension of 
quantized current algebras.  The composed currents appear in 
this approach from the braid group action. 

One can develop the same approach for the Yangian double $\aDY{\fg}$ with 
$\ggo=\gga,\goN,\gsp$ and
restricted to the category of highest weight representations. 
In this context,  the simple root currents are viewed as generating series 
in the completed subalgebra 
$\oY_F$ (see e.g. \cite{DKh,DKhP}) and their products  
are considered as rational functions of the complex parameters, with simple poles 
defined by the zeros of the functions entering their commutation relations.
The composed currents are then defined  as the residues at these poles.
 
The calculations in this section  are the extension of the calculations \cite{DKh}
to the Yangian doubles $\aDY{\fg}$ of all classical series.

Fix a positive integer $n\geq2$. 
In order to describe Yangian doubles  
simultaneously for all $\ggo$ of the classical series
we introduce two integer parameters $\frd_\ggo$ and $\epsi_\fg$ 
\begin{equation}\label{prime-par}
\frd_\ggo=\begin{cases}
n+1,&\mathfrak{g}=\mathfrak{sl}_{n}\,,\\
0,&\mathfrak{g}=\mathfrak{o}_{2n+1}\,,\\
1,&\mathfrak{g}=\mathfrak{sp}_{2n},\ \mathfrak{o}_{2n}\,,
\end{cases}\qquad 
\epsi_{\ggo}=\begin{cases}
0,& \ggo=\gga\,,\\
1,&\ggo=\ggb,\ \ggd\,,\\
-1,& \ggo=\gsp\,.
\end{cases}
\end{equation}
The first parameter is used to introduce the indices for the numbering 
composed currents in the Yangian double, while the second parameter 
allows to distinguish linear, orthogonal and symplectic Yangian doubles.

For any integer $i$ we define the map $i\to i^\prime$ as 
\begin{equation}\label{prime}
i^{\,\prime}=\frd_\ggo-i\,.
\end{equation}
This map acts invariantly on the set  
\begin{equation}\label{I-g}
\Ig=\{n',n'+1,\ldots,n-1,n\}\,.
\end{equation}
 Denote by $N_\ggo$ the cardinality of the set $\Ig$. 
It is easy to verify that 
\begin{equation}\label{N-g}
N_\ggo=|\Ig|\equiv N=2n+1-\frd_\ggo\,.
\end{equation} 
Explicitly, $N_\ggo$ is equal to $n$, $2n+1$, $2n$ and $2n$
for $\ggo=\gga$, $\ggb$, 
$\gsp$, $\ggd$  respectively, and coincides with the dimension of the 
vector representation of $\ggo$.

\subsection{Definition of composed currents}

Definition of composed currents for the Yangian doubles of all classical series 
was given in \cite{LP1}. Here we present their definition without proofs.  

For  $\ggo\not=\gga$, let $\bvphi_\ggo$ and $\theta_\fg$ be the parameters 
\begin{equation}\label{vphi}
\bvphi_\fg=2\,\frd_\fg+(\epsi_\fg-1)/2,\qquad 
\theta_\fg=\epsi_\fg+(1-\frd_\fg)/2\,.
\end{equation}
Using these parameters we define for each node of the Dynkin diagram 
of the classical series
the shifted parameters $z_a$ for the algebra
$\fg$ of the  series $B_n$, $C_n$, and $D_n$  
\begin{equation}\label{z-sh1}
 z_a=\begin{cases}
 z,&0\leq a<\bvphi_\fg\,,\\
 z-c(a+\theta_\fg),& \bvphi_\fg\leq a\leq n-1\,.
 \end{cases}
 \end{equation}

For $\ggo\not=\gga$, using the map \eqref{prime}, we introduce {\sl auxiliary} 
simple root currents  \cite{LP1}  by the equality 
\begin{equation}\label{dep-cur}
F_{a'-1}(z)=-\ F_a(z_a),\qquad \bvphi_\ggo\leq a\leq n-1\,.
\end{equation}

Combining simple roots and (for $\ggo\not=\gga$) auxiliary simple roots currents 
we can define 
the {\sl composed} currents $F_{j,i}(u)$, $i,j\in\Ig$. 
For the algebras $\ggo=\gga$, $\ggb$, and $\ggc$ they read 
\begin{equation}\label{ccABC}
F_{j,i}(u)=F_i(u)\, F_{i+1}(u)\cdots F_{j-2}(u)\, F_{j-1}(u)\,,
\qquad n'\leq i<j\leq n
\end{equation} 
while for the algebra $\ggo=\ggd$ they take the form
\begin{equation}\label{ccD}
F_{j,i}(u)=\begin{cases}
F_i(u)\, F_{i+1}(u)\cdots  F_{j-1}(u),\quad 2\leq i<j\leq n
\quad\mbox{or}\quad
2\leq j'<i'\leq n\,,\\
-\ F_{i}(u)\cdots F_{-2}(u)\, F_{j}(u),\quad 2\leq i'\leq n,\quad j=0,1\,,\\
0,\quad i=0,\quad j=1\,,\\
 F_{i}(u)\, F_{2}(u)\cdots F_{j-1}(u),\quad i=0,1,\quad 2\leq j\leq n\,,\\
-\ \Big(F_i(u)\cdots F_{-2}(u)\Big)\, F_0(u)\, F_1(u)\, \Big(F_2(u)\cdots F_{j-1}(u)\Big),\quad
2\leq i',j\leq n\,.
\end{cases}
\end{equation}
The composed current $F_{2,-1}(u)=-F_0(u)\,F_1(u)$ 
according to the last line in \r{ccD} for 
$i'=j=2$ is equal (up to the sign) to the product of the commuting 
currents $F_0(u)$ and $F_1(u)$.

\subsection{Properties of composed currents for $\aDY{\gga}$ and $\aDY{\ggd}$}

The composed currents $F_{j,i}(u)$ are labeled 
by two indices $i,j\in \Ig$ such that $i<j$. The simple root currents 
$F_a(u)$ for $a\in\PSR_{\fg}$ is accordingly denoted as $F_{a+1,a}(u)$. 
In the appendix~\ref{AppA}, following the results obtained in \cite{DKh},
we construct the composed current $F_{3,1}(u)= F_1(u)\,F_2(u)=
F_{2,1}(u)\,F_{3,2}(v)$ in the analytical context
and show that it is a well-defined element of the completed subalgebra 
$\oY_F$ in the Yangian double $\aDY{\fs\fl_3}$. 

To lighten the presentation we will use the notation 
\begin{equation*}
f_a(u,v)\equiv f(u,v)=\frac{u-v+c}{u-v},\qquad h_a(u,v)\equiv h(u,v)=\frac{u-v+c}{c}
\end{equation*}
for the roots $\rt_a$ such that $(\rt_a,\rt_a)=2$.

To investigate the connection between the Serre relations and the properties 
of the composed currents for the Yangian double $\aDY{\gga}$, 
it is sufficient to consider the simplest nontrivial case $\aDY{\fs\fl_3}$. 
We start from the commutation relation between 
$F_{3,1}(u)$ and $F_2(v)$
and  use the Serre relation 
\begin{equation}\label{Ac4}
\mathop{\rm Sym}\limits_{v_1,v_2}\Big[F_2(v_1),[F_2(v_2),F_{1}(u)]\Big]=0
\end{equation}
to calculate these commutation relations.  To do this we rewrite the 
Serre relation \r{Ac4} in the form 
\begin{equation}\label{Ac5}
\mathop{\rm Sym}\limits_{v_1,v_2}
\Big(F_2(v_1)F_2(v_2)F_1(u)-2F_2(v_1)F_1(u)F_2(v_2)+
F_1(u)F_2(v_1)F_1(v_2)\Big)=0\,.
\end{equation}
We will  transform the linear combination of the products of the currents 
under symmetrization 
over $v_1$ and $v_2$ in \r{Ac5} moving the currents $F_1(u)$ 
to the left using the commutation relation \r{Ac2}. 

Doing this we rewrite the combination 
\begin{equation}\label{Ac6}
F_2(v_1)F_2(v_2)F_1(u)-2F_2(v_1)F_1(u)F_2(v_2)+F_1(u)F_2(v_1)F_2(v_2)
\end{equation}
in the form 
\begin{equation}\label{Ac7}
\begin{split}
&h(v_1,v_2)g(u,v_1)g(u,v_2)\,F_1(u)F_2(v_1)F_2(v_2)+\\
&\qquad+
c\delta(u,v_2)F_2(v_1)F_{3,1}(u)- c\delta(u,v_1) f(u,v_2)F_{3,1}(u)F_2(v_2)\,.
\end{split}
\end{equation}
Symmetrizing all expressions in \r{Ac7} over $v_1$ and $v_2$ and using 
the commutation relation 
\begin{equation}\label{Ac8}
h(v_1,v_2)F_2(v_1)F_2(v_2)+h(v_2,v_1)F_2(v_2)F_2(v_1)=0
\end{equation}
  and the linear independence of the $\delta$-functions 
one concludes that the Serre relation \r{Ac4} or \r{Ac5} 
yields the commutation relation 
\begin{equation}\label{Ac9}
F_2(v)F_{3,1}(u)= f(u,v) F_{3,1}(u)F_2(v)\,.
\end{equation}

Analogously, from the second Serre relation 
\begin{equation}\label{Ac42}
\mathop{\rm Sym}\limits_{u_1,u_2}\Big[F_1(u_1),[F_1(u_2),F_{2}(v)]\Big]=0
\end{equation}
one can get the commutation relation 
\begin{equation}\label{Ac92}
F_{3,1}(u)F_1(v)= f(v,u) F_1(v)F_{3,1}(u)\,.
\end{equation}
Let us remind that in \eqref{Ac9}, $f(u,v)$ is considered as a series in $v/u$, while in \eqref{Ac92}, $f(v,u)$ is considered as a series in $u/v$.

For $n>3$ the composed currents $F_{j,i}(u)$ for $1\leq i<j\leq n$ 
\begin{equation}\label{comp-cu}
F_{j,i}(u)=F_i(u)\, F_{i+1}(u)\cdots F_{j-2}(u)\, F_{j-1}(u)
\end{equation}
can be defined inductively in  the subalgebra 
$\oY_F$.
The commutation relations between all composed and simple root currents 
can be formulated as 
\begin{lemma}\label{lem21}
The following relations hold in $\oY_F$ for any $i<j$ and $k<l$:
\begin{eqnarray}
F_{j,i}(u)\,F_{l,k}(v)&=&F_{l,k}(v)\,F_{j,i}(u)\,,\quad j<k\,,\label{cc1}\\
F_{j,i}(u)\,F_{l,k}(v)&=&F_{l,k}(v)\,F_{j,i}(u)\,,\quad i<k<l<j\,,\label{cc5}\\
h(v,u)\,F_{j,i}(u)\,F_{l,k}(v)&=&g(v,u)^{-1}\,F_{l,k}(v)\,F_{j,i}(u)\,,\quad 
i<j=k<l\,,
\label{cc2}\\
f(u,v)\,F_{j,i}(u)\,F_{l,k}(v)&=&
f(v,u)\,F_{l,k}(v)\,F_{j,i}(u)\,,\quad j=l,\ i=l\,,\label{cc4}\\
F_{j,i}(u)\,F_{l,k}(v)&=&f(v,u)\,F_{l,k}(v)\,F_{j,i}(u)\,,\quad k<i,\ j=l\,,
\label{cc3}\\
f(u,v)\,F_{j,i}(u)\,F_{l,k}(v)&=&
F_{l,k}(v)\,F_{j,i}(u)\,,\quad i=k,\ j<l\,.\label{cc6}
\end{eqnarray}
The rational functions $f(u,v)$  (resp. $f(v,u)$) in the l.h.s. (resp. the r.h.s.) of the commutation relations \r{cc4}--\r{cc6} should be 
understood as power series of the ratio $v/u$ (resp.  the ratio $u/v$).
\end{lemma} 

The lemma is the analog of the proposition~A.1 in the paper
\cite{KhP}, where it was proved for the quantum affine algebra 
$U_q(\widehat{\mathfrak{gl}}_n)$. 
The commutation relations \r{cc1}--\r{cc4}  are 
consequences of the commutation relations \r{FFalt}
between simple root currents. The commutation relations 
\r{cc3} and \r{cc6} are consequence of the Serre relations \r{SR-gF} as it was shown 
above for the algebra $\aDY{\fs\fl_3}$. 
An immediate corollary of the lemma~\ref{lem21}
is the following corollary. 
\begin{cor}\label{cor-intl}
The commutation relations 
\r{inla2} and \r{inla2a} for the composed currents 
with interlaced indices
\begin{eqnarray}
h(v,u)\,f(u,v)\,F_{j,i}(u)\,F_{l,k}(v)&=&g(v,u)^{-1}\,F_{l,k}(v)\,F_{j,i}(u),
\quad i<k<j<l\,,\label{inla2}\\
g(u,v)^{-1}\,F_{j,i}(u)\,F_{l,k}(v)&=&h(u,v)\,f(v,u)\,F_{l,k}(v)\,F_{j,i}(u),\quad
k<i<l<j\label{inla2a}
\end{eqnarray}
are consequence of the commutation relations \r{cc1}--\r{cc6}.
\end{cor}

Let us prove this corollary for the relation  \r{inla2}. The relation \r{inla2a} can 
be proved  similarly. 
First note that, according 
to their definitions \r{comp-cu},  the composed currents in the 
commutation relation \r{inla2} 
can be presented in the factorized form  
\begin{equation}\label{inla1}
F_{j,i}(u)=F_{k,i}(u)\,F_{j,k}(u),\qquad
F_{l,k}(v)=F_{j,k}(v)\,F_{l,j}(v),\qquad i<k<j<l\,.
\end{equation}

 The commutation relation \r{inla2} 
follows from the following chain of equalities 
\begin{equation}\label{inla3}
\begin{split}
&h(v,u)\,f(u,v)\,F_{j,i}(u)\,F_{l,k}(v)\mathop{=}\limits^{\r{inla1}}
h(v,u)\,f(u,v)\,F_{j,i}(u)\,F_{j,k}(v)\,F_{l,j}(v)\\
&\quad\mathop{=}\limits^{\r{cc3}}\ 
h(v,u)\,F_{j,k}(v)\,F_{j,i}(u)\,F_{l,j}(v)\mathop{=}\limits^{\r{inla1}}
h(v,u)\,F_{j,k}(v)\,F_{k,i}(u)\,F_{j,k}(u)\,F_{l,j}(v)\\
&\quad\mathop{=}\limits^{\r{cc2}}\ 
g(v,u)^{-1}\,F_{j,k}(v)\,F_{k,i}(u)\,F_{l,j}(v)\,F_{j,k}(u)\\
&\quad 
\mathop{=}\limits^{\r{cc5}}
g(v,u)^{-1}\,F_{j,k}(v)\,F_{l,j}(v)\,F_{k,i}(u)\,F_{j,k}(u)
\mathop{=}\limits^{\r{inla1}}
g(v,u)^{-1}\,F_{l,k}(v)\,F_{j,i}(u)\,.
\end{split}
\end{equation}

As it was argued in \cite{DKh,KhP}, the commutation relations between 
currents given by the lemma~\ref{lem21} restricted to the category of 
the highest weight representations are  relations 
between  functions of the parameters $u$ and $v$ which 
do not have zeros or poles. It implies that 
if one side of the commutation relation between currents 
$F(u)$ and $F'(v)$ is a product $G(u,v)\,F(u)\,F'(v)$ with some rational function 
$G(u,v)$ then the product of the
currents $F(u)\,F'(v)$  has a pole 
where the function $G(u,v)$ has zero and has a zero where 
the function $G(u,v)$ has pole. For example, it follows 
from the commutation relations \r{cc1} and \r{cc5} that the corresponding products
as functions of $u$ and $v$ have no zeros neither poles. 
The relation \r{cc4} means that the product of the same composed 
currents $F_{j,i}(u)\,F_{j,i}(v)$ have a simple pole when $v=u+c$ and 
simple zero when $u=v$, etc.  

The commutation relations 
\r{inla2} and \r{inla2a} imply that  the product 
of the composed currents $F_{j,i}(u)\,F_{l,k}(v)$ 
with interlaced indices $i<k<j<l$
has two simple poles at the points $u=v\pm c$ and one simple zero at $u=v$
while the inverse product $F_{l,k}(v)\,F_{j,i}(u)$ has one simple pole at $u=v$. 
Similarly, for $k<i<l<j$ the product 
of the composed currents $F_{j,i}(u)\,F_{l,k}(v)$ has a simple pole at $u=v$ and 
the inverse product $F_{l,k}(v)\,F_{j,i}(u)$ has two simple poles at the points 
$u=v\pm c$ and one simple zero at $u=v$.
The importance to investigate the 
commutation relation between composed currents
with interlaced indices was mentioned in 
\cite{KT2}.


The algebra of the composed currents in the Yangian 
doubles  $\aDY{\ggb}$, $\aDY{\gsp}$, and $\aDY{\ggd}$ 
may be similarly investigated. 
The commutation relations of the composed currents $F_{j,i}(u)$ 
in these cases for the values of the indices 
$\varphi_\ggo\leq i<j\leq n$ and $\varphi_\ggo\leq j'<i'\leq n$
repeat the $\gga$ type structure shown in the lemma~\ref{lem21}. 
Due to the simply laced character of the algebra $\fo_{2n}$ 
the algebra of the composed currents \r{ccD} in the Yangian double 
$\aDY{\ggd}$ maybe investigated in the same way as for the Yangian 
$\aDY{\gga}$. We left this as an exercise for the interested readers.

Examples of composed currents  which  
are not of $\gga$ type and their relations
to the corresponding Serre relations are given in sections~\ref{AppF}
and \ref{AppG} for the Yangian doubles $\aDY{\fs\fp_4}$ and $\aDY{\fo_5}$.

\subsection{Composed currents for  $\aDY{\gsp}$}\label{AppF}

The composed currents for the Yangian double $\aDY{\gsp}$ are defined 
by the equalities \r{ccABC}. In order to describe the commutation relations between 
them we introduce  simplified notations for the functions 
$h_0(u,v)$ and $f_0(u,v)$
\begin{equation*}\label{fhc}
h_0(u,v)\equiv \hc(u,v)=\frac{u-v+2\,c}{c},\quad 
f_0(u,v)\equiv \fc(u,v)=\frac{u-v+2\,c}{u-v},\quad  
\fc(u,v)=\hc(u,v)\,g(u,v)\,.
\end{equation*}
The currents $F_i(u)$ for $i=1,\ldots,n-1$ form the $\mathfrak{sl}_n$-type
completed  subalgebra $\oY_F$ in $\aDY{\gsp}$. 
Thus, the corresponding composed currents can be 
deduced from the study of the $\fs\fl_n$ case.
Moreover, since the currents $F_i(u)$ for $i=2,\ldots,n-1$
commute with the special current $F_0(v)$, the 
only new feature comes from the currents $F_0(u)$ and $F_1(u)$. Hence, in order to 
describe the composed currents in the Yangian double $\aDY{\gsp}$
 it is sufficient to consider the case $n=2$, namely $\aDY{\fs\fp_4}$.
\r{FFalt}
The only nontrivial commutation relation between neighboring simple root 
currents \r{FFalt} in this case will be 
\begin{equation}\label{C3}
(v-u+2\,c)\ F_0(u)F_1(v)=(v-u)\ F_1(v)F_0(u)
\end{equation}
which can be interpreted 
as an analytical continuation of the 
products of currents from the domain $|u|\gg|v|$ to 
the domain  $|v|\gg|u|$. 
These analytical properties allow to rewrite the commutation relation 
\r{C3} in the form 
\begin{equation}\label{C4}
F_1(v)F_0(u)=\fc(v,u)F_0(u)F_1(v)+2\,c\,\delta(u,v)F_{2,0}(u)
\end{equation}
where the rational function $\fc(v,u)$ in the l.h.s. of \r{C4} should be understood as 
a series with respect to the nonnegative powers of $v/u$. 
The element 
$F_{2,0}(u)$ is the  composed current and is equal to the residue 
of the product $F_1(v)F_0(u)$ at the point $v=u$
\begin{equation}\label{C5}
\begin{split}
&2\,c\ F_{2,0}(u)=\mathop{\rm res}\limits_{v=u}F_1(v)F_0(u)=
(v-u)F_1(v)F_0(u)\Big|_{v=u}=\\
&\qquad\qquad=
(v-u+2\,c)F_0(u)F_1(v)\Big|_{v=u}
=2\,c\ F_0(u)F_1(u)\,.
\end{split}
\end{equation}
This is a well-defined object in the Yangian double restricted to the category 
of  highest weight representations as it was explained in appendix~\ref{AppA}.

Analogously  the product $F_0(u)F_1(v)$
can be analytically continued from the domain $|u|\gg |v|$ to the domain 
$|u|\ll |v|$ by rewriting 
the same commutation relation \r{C3}  in the form 
\begin{equation}\label{C6}
F_0(u)F_1(v)=\fc(u-2\,c,v)F_1(v)F_0(u)+2\,c\,\delta(u-2\,c,v)F_{1,-1}(v+2\,c)
\end{equation}
where the rational function $\fc(u-2\,c,v)$ in the r.h.s. of \r{C6} should be understood 
as a series over nonnegative powers of $(u-2\,c)/v$. 
The element 
$F_{1,-1}(v)$ is another composed current which is equal to the residue 
of the product $F_0(u)F_1(v)$ at the point $u=v+2\,c$
\begin{equation}\label{C7}
\begin{split}
&2\,c\ F_{1,-1}(v+2\,c)=\mathop{\rm res}\limits_{u=v+2\,c}F_0(u)F_1(v)=
(v-u+2\,c)F_0(u)F_1(v)\Big|_{u=v+2\,c}=\\
&\qquad=(v-u)F_1(v)F_0(u)\Big|_{u=v+2\,c}=-\ 2\,c\, F_1(v)F_0(v+2\,c)\,.
\end{split}
\end{equation}
Again one can show that this element is a well-defined object
in the Yangian double restricted to the category of highest weight 
representations. 

The  shift by $2\,c$ and the indices in the composed current
$F_{1,-1}(v+2\,c)$ \r{C7} 
can be explained as follows. Recall  the auxiliary simple root current 
$F_{-1}(u)=-F_1(u-2\,c)$ defined by \r{ccABC}. In the notation 
with two indices it can be written as $F_{0,-1}(u)= - F_{2,1}(u-2\,c)$. 
Then the composed current $F_{1,-1}(u)$ defined by \r{C7}
can be written as 
\begin{equation}\label{C8}
 F_{1,-1}(u)=-F_1(u-2\,c)\,F_0(u)=F_{-1}(u)\,F_0(u)=F_{0,-1}(u)\,F_{1,0}(u)\,.
 \end{equation}
 So both composed currents \r{C5} and \r{C7} can be described by
  the single formula
 \begin{equation}\label{C9}
 F_{i+1,i-1}(u)=F_{i-1}(u)\,F_{i}(u)=F_{i,i-1}(u)\,F_{i+1,i}(u),\quad i=0,1.
 \end{equation}

Now we can address the question of the commutation relations 
between composed and simple root currents. 
Using the relation $F_{1,0}(v)^2=0$, which comes from the commutation relations \r{FFalt}, 
 one finds  
\begin{equation}\label{C16}
F_{2,0}(v)F_{1,0}(u)= \fc(u,v) F_{1,0}(u) F_{2,0}(v)
\end{equation}
and 
\begin{equation}\label{C16-1}
F_{1,0}(v)F_{1,-1}(u)= \fc(u,v) F_{1,-1}(u) F_{1,0}(v)\,.
\end{equation}

Now we can compute the commutation relation between the composed 
current $F_{2,0}(u)$ and the simple root current 
$F_{2,1}(v)=F_1(v)$.  To do this, we consider the product 
\begin{equation*}
f(u,v)\,F_{2,0}(u)\,F_{2,1}(v)=f(u,v)\,F_{0}(u)\,F_{1}(u)\,F_{1}(v)
\end{equation*}
and find that 
  \begin{equation}\label{C11a}
  f(u,v)\,F_{2,0}(u)F_{2,1}(v)=f(u-2\,c,v)\,F_{2,1}(v)F_{2,0}(u)+
  c\,\delta(u-2\,c,v)\,F_{2,-1}(u)
  \end{equation}
where now the $\delta$-function term is proportional to a new 
composed current $F_{2,-1}(u)$ given by 
\begin{equation}\label{C10}
F_{2,-1}(u)=-\ F_1(u-2c)\,F_0(u)\,F_1(u)\,.
\end{equation}

Multiplying the commutation relation \r{C11a} by the function 
$\hc(v,u)$ which have a simple zero at $v=u-2\,c$, we get rid off 
the $\delta$-function term and obtain the commutation relation 
 \begin{equation}\label{C11d}
  \hc(v,u)\,f(u,v)\,F_{2,0}(u)F_{2,1}(v)=h(v,u)\,F_{2,1}(v)F_{2,0}(u)\,.
  \end{equation}
The r.h.s. of this commutation relation implies that 
the product $F_{2,1}(v)F_{2,0}(u)$
has a simple pole at $v=u-c$. 
The l.h.s. of the same commutation relation implies that 
the product $F_{2,0}(u)F_{2,1}(v)$ has 
two simple poles at $v=u-2\,c$, $v=u+c$ and one simple zero 
at $u=v$. 

Using a similar approach we can calculate the commutation relations 
between the new composed current $F_{2,-1}(u)$ and the simple root 
currents $F_{2,1}(v)=F_1(v)$ and $F_{1,0}(v)=F_0(v)$. 
We obtain 
\begin{equation}\label{C51}
f(u,v)\,F_{2,-1}(u)\,F_{2,1}(v)=f(v,u-2\,c)\,F_{2,1}(v)\,F_{2,-1}(u)
\end{equation}
and 
\begin{equation}\label{C51a}
F_{2,-1}(u)\,F_{1,0}(v)=F_{1,0}(v)\,F_{2,-1}(u)\,.
\end{equation}

One can demonstrate that the commutation relation \r{C16} is equivalent to 
the Serre relation 
\begin{equation*}\label{C13}
\mathop{\rm Sym}\limits_{u_1,u_2}\Big[F_0(u_1),[F_0(u_2),F_1(v)]\Big]=0
\end{equation*}
while the commutation relation  \r{C51} is equivalent to
the Serre relation 
\begin{equation*}\label{C13a}
\mathop{\rm Sym}\limits_{v_1,v_2,v_3}
\Big[F_1(v_1),\big[F_1(v_2),[F_1(v_3),F_0(u)]\big]\Big]=0\,.
\end{equation*}

A special treatment is needed to investigate the commutation relation 
of composed currents with interlaced indices. In $\aDY{\fs\fp_4}$,
it corresponds to the composed currents 
$F_{1,-1}(u)=-F_1(u-2\,c)\,F_0(u)$ and $F_{2,0}(v)=F_0(v)\,F_1(v)$. 
Their commutation relation is 
\begin{equation}\label{inlC2}
h(u-2\,c,v)\,\fc(u,v)\,F_{1,-1}(u)\,F_{2,0}(v)=
g(u,v)^{-1}\,f(v,u-2\,c)\,F_{2,0}(v)\,F_{1,-1}(u)\,.
\end{equation}

The commutation relation \r{inlC2} tells that the product $F_{1,-1}(u)\,F_{2,0}(v)$, considered 
as a  function in the category of  highest weight representations 
of $\aDY{\fs\fp_4}$, has  simple poles at the points 
$u=v+c$, $u=v-2\,c$ and a simple zero at the point $u=v$. 
Analogously, the product 
$F_{2,0}(v)\,F_{1,-1}(u)$ has simple poles at the points $u=v$, $u=v+3\,c$ and 
a simple zero at the point $u=v+2\,c$.

\subsection{Composed currents for  $\aDY{\ggb}$}\label{AppG}

The composed currents are defined by the equalities \r{ccABC}. To 
find their commutation relations we follow  
the same approach as in the section~\ref{AppF}.

To understand how one can get these commutation relations and which 
of them are equivalent to the Serre relations 
 it is sufficient to consider the case of the
Yangian double $\aDY{\fo_5}$. The general case 
is then  obvious in view of the embedding 
of $\aDY{\fo_{2n-1}}$ into $\aDY{\ggb}$. 

The composed currents $F_{j,i}$ of the first level (i.e. with $j-i=2$) 
can be defined by rewriting the commutation relations for the simple 
root currents $F_{1,0}(u)$ and $F_{2,1}(v)$ in the analytical
language 
\begin{equation}\label{B4}
\begin{split}
F_1(v)F_0(u)&=f(v,u)F_0(u)F_1(v)+c\delta(u,v)F_{2,0}(u)\,,\\
F_0(u)F_1(v)&=f(u,v+c)F_1(v)F_0(u)+c\delta(u,v+c)F_{0,-2}(v+c/2)\,,\\
F_0(u)F_0(v)&=f(v,u+c/2)\ F_0(v)F_0(u)+c\delta(u,v-c/2)\ F_{1,-1}(v-c/2)\,.
\end{split}
\end{equation} 

Next one can address the question  of the commutation relations 
of the composed currents $F_{2,0}(v)$ and $F_{0,-2}(v)$
with the simple root current $F_0(u)$. 
One gets 
\begin{equation}\label{B21}
\begin{split}
F_{2,0}(v)F_0(u)&=f(v,u)f(u,v+c/2)F_0(u)F_{2,0}(v)-c\delta(v,u-c/2)F_{2,-1}(v)\,,\\
F_0(u)F_{0,-2}(v)&=f(u,v+c/2)f(v,u)F_{0,-2}(v)F_{0}(u)+c\delta(v,u)F_{1,-2}(v)\,,
\end{split}
\end{equation}
where the two new  composed currents of the next level 
$F_{1,-2}(v)=F_1(v-c/2)F_0(v+c/2)F_0(v)$ and 
$F_{2,-1}(v)=-F_0(v+c/2)F_0(v)F_1(v)$  are  well-defined elements in the Yangian double $\aDY{\fo_5}$ in the sense explained 
in appendix~\ref{AppA}.  

The commutation relations 
between the composed currents $F_{2,0}(v)$ and $F_{0,-2}(v)$
 and the simple root current 
$F_{2,1}(u)$ are 
\begin{equation}\label{B16}
\begin{split}
F_{2,1}(v)\,F_{2,0}(u)&= f(u,v)\, F_{2,0}(u)\,F_{2,1}(v)\,,\\
F_{0,-2}(v)\,F_{-1,-2}(u)&= f(u,v)\, F_{-1,-2}(u)\,F_{0,-2}(v)\,,
\end{split}
\end{equation}
and are equivalent to the Serre relation 
\begin{equation}\label{B13}
\mathop{\rm Sym}\limits_{u_1,u_2}\Big[F_1(u_1),[F_1(u_2),F_0(v)]\Big]=0\,.
\end{equation}
On the other hand the commutation relations 
between the composed current $F_{1,-2}(v)$ or $F_{2,-1}(u)$
and the simple root current 
$F_{1,0}(u)=F_{0}(u)$ or $F_{0,-1}(v)=-F_0(v+c/2)$
\begin{equation}\label{B35}
\begin{split}
F_{2,-1}(u)F_{0,-1}(v)&=f(v,u)F_{0,-1}(v)F_{2,-1}(u)\,,\\
F_{1,0}(u)F_{1,-2}(v)&=f(v,u)\ F_{1,-2}(v)F_0(u)
\end{split}
\end{equation}
are equivalent to the Serre relation 
\begin{equation}\label{B24}
\mathop{\rm Sym}\limits_{v_1,v_2,v_3}
\Big[F_0(v_1),\big[F_0(v_2),[F_0(v_3),F_1(u)]\big]\Big]=0\,.
\end{equation}

Finally, the commutation relations which creates 
 the composed current of the highest level in Yangian double $\aDY{\fo_5}$
\begin{equation}\label{B25}
F_{2,-2}(u)=F_1(u-c/2)F_0(c+c/2)F_0(u)F_1(u)=F_{-2}(u)F_{-1}(u)F_{0}(u)F_{1}(u)\,.
\end{equation}
are 
\begin{equation*}\label{B50a}
\begin{split}
f(v,u-c/2)F_{2,1}(v)F_{1,-2}(u)&=f(v,u)F_{1,-2}(u)F_{2,1}(v)+c\delta(u,v)F_{2,-2}(v)\,,\\
f(u,v-c/2)F_{2,-1}(u)F_{-1,-2}(v)&
=f(u,v)F_{-1,-2}(v)F_{2,-1}(u)+c\delta(u,v)F_{2,-2}(u)\,,\\
f(u,v-c/2)F_{2,0}(u)F_{0,-2}(v)&= 
f(u,v+c/2)f(v,u)F_{0,-2}(v)F_{2,0}(u)-c\delta(u,v)F_{2,-2}(v)\,.
\end{split}
\end{equation*}

\subsection{Additional zeros in the product of the simple roots currents}

It was shown in the section~\ref{an-pro-BV} that the Serre relations lead 
to the apparition of additional zeros in the product of the currents. 
Let us demonstrate how these additional zeros imply the commutation
relations of the simple and composed roots currents on the example of the 
Yangian double $\aDY{\fo_5}$.

Consider for example the  first commutation relation in \r{B16} 
which is equivalent to the 
Serre relation \r{B13}. The left hand side of this commutation relation tells us that 
the product $F_1(v)F_0(u)F_1(u)$ has no zeros nor poles. But 
individual commutation relations between the simple root currents yields 
two simple poles at the points $v=u$, $v=u-c$ and one simple 
zero at the point $v=u$, leading to a potential pole at $v=u-c$ for the product $F_1(v)F_0(u)F_1(u)$. 
Since this product has 
no zeros nor poles, it implies that it should have {\sl additional} (i.e. not coming from the commutation relations)
zero at the point 
$v=u-c$,  which is a consequence of the Serre relation \r{B13}. 

The r.h.s. of the commutation relation \r{B16} yields that 
the product $F_0(u)F_1(u)F_1(v)$ should 
have a simple pole at the point $u=v-c$ and a simple zero at $u=v$. 
Indeed, the pole at $u=v-c$ follows from the product $F_1(u)F_1(v)$,
but the product $F_0(u)F_1(v)$ produces the pole at the point $u=v+c$
which is absent in the r.h.s. of the commutation relation \r{B16}. It means 
that the product $F_0(u)F_1(u)F_1(v)$ should have an additional zero 
when $u=v+c$ as a consequence of the Serre relation \r{B13}. 
Concluding, we may say that, as a consequence of the Serre relations \r{B13},
the products $F_1(v)F_0(u)F_1(u)$ and $F_0(u)F_1(u)F_1(v)$  have 
an additional zero at the point $u=v+c$. This is 
indeed true and was proved above using results of the paper \cite{E}.

In the l.h.s. of the first commutation relation in \r{B35} one has the product 
 \begin{equation*}
 F_0(v+c/2)F_0(v)F_1(v)F_0(u)
 \end{equation*} 
 which 
has simple poles at the points $u=v+c$, $u=v+c/2$ and $u=v$. On the other 
hand this product has simple zeros at the points  $u=v+c/2$ and $u=v$.
Since according to \r{B35} this product has no zeros nor poles, it means 
that it has an additional zero at the point $u=v+c$. Now in the r.h.s. 
of the commutation relation \r{B35} one has the product 
$F_0(u)F_0(v+c/2)F_0(v)F_1(v)$ which has simple
poles at the points $u=v$, $u=v-c/2$, $u=v+c$ and simple 
zeros at the points $u=v$ and $u=v+c/2$ as result of the commutation 
relations between simple roots currents. But the r.h.s. of \r{B35} tell us that 
the product $F_0(u)F_0(v+c/2)F_0(v)F_1(v)$ should have only simple 
pole at the point $u=v-c/2$ and a simple zero at the point $u=v+c/2$. 
In order to get this, one has to request that the product 
$F_0(u)F_0(v+c/2)F_0(v)F_1(v)$ has an additional zero at the point $u=v+c$
in order to compensate the pole in the same point.

\section{Discussion}\label{disc}

In the present paper, we established the connexion between the Serre relations and the commutation relations of composed currents, within 
the category of highest weight representations. The key feature is the use of a normal order between half-currents, compatible with the notion of highest weight vectors. Clearly, one could proceed in the same way in the category of lowest weight representations. In that case the normal order will be different, so as to be compatible with the notion of lowest weight vectors.

Within the framework of the algebraic Bethe ansatz \cite{TF79,S22}, Bethe vectors in a generic $\fg$-invariant integrable model are constructed from the entries of the monodromy matrix, which satisfy the same commutation relations as the $T$-operator of the Yangian associated with a finite-dimensional representation of $Y(\fg)$. Although the set-up is clear, the explicit construction of Bethe vectors appears to be technically involved, so that one seeks for new ways to tackle the problem. 

For quantum affine algebras, viewed as quantum doubles of their Borel subalgebras, it was shown in \cite{EKhP,KhP} that the construction becomes simpler in the context of quantum doubles. It led to the development of the projection method. The terminology reflects the fact that, within this approach, off-shell Bethe vectors are obtained as projections of products of currents onto intersections of Borel subalgebras of different types in the quantum affine algebra. Subsequently, the projection method was extended in \cite{HLPRS17} to integrable models associated with the supersymmetric Yangian doubles $\mathcal{D}Y(\fg\fl(m|n))$.

By analogy, one may realize off-shell Bethe vectors in a generic $\fg$-invariant integrable model as projections of products of currents in the Yangian double $\aDY{\fg}$. To establish that the resulting vectors indeed satisfy the defining properties of off-shell Bethe vectors formulated in \cite{LPR-RR}, it is necessary to investigate the structure of products of currents and to understand the role played by the Serre relations in these products.

This constitutes the main objective of the present paper. The results obtained here provide the necessary groundwork for proving that projections of products of currents in the Yangian double $\aDY{\fg}$ yield a valid description of off-shell Bethe vectors in generic $\fg$-invariant integrable models. This application will be developed in the forthcoming work \cite{LPR-BV-CU}.

\section*{Acknowledgement}
S.P. acknowledges support from the PAUSE Programme and the hospitality of LAPTh, where this work was conducted.
The work of A.~L. was supported by the Beijing Natural Science Foundation (IS24006) and Beijing
Talent Program.

\appendix

\section{Composed currents as well-defined objects in $\oY_F$
\label{AppA}}

In the section~\ref{cc-property}
we  consider  the composed  currents in the analytical framework 
following the ideas developed in the paper \cite{DKh}. 
Namely, we demonstrate that the product of the 
currents $F_a(u)\,F_b(v)$ 
being restricted to the category of the 
highest weight representations
has zeros (poles) at the points 
where function $\cre_{a,b}(u,v)$ have poles (zeros). 
To understand this approach, it is sufficient to consider the 
case of the Yangian double $\aDY{\fs\fl_3}$. The general case for
arbitrary Yangian double $\aDY{\fg}$ can be treated 
similarly. 

There is only one 
 commutation relation \r{FFalt} for the different simple root currents
 \begin{equation}\label{Ac1}
(v-u+c)\ F_1(u)F_2(v)=(v-u)\ F_2(v)F_1(u)\,.
\end{equation}
 According to \cite{DKh,E} 
 the product   
$F_2(v)F_1(u)$ is an 
analytical function in the domain $|v|\gg |u|$ with a simple pole in the point $v=u$. 
On the other hand  
the product  $F_1(u)F_2(v)$ is an analytical function in the domain 
$|u|\gg |v+c|$ with a simple pole at $u=v+c$. 
Then  the commutation relation \r{Ac1} between simple root currents 
can be interpreted  as a method of analytical 
continuation of the 
products of the currents from one domain to another.

In \cite{DKh} these analytical properties were rewritten in the form of  
commutation 
relations. In order to analytically continue the product $F_2(v)F_1(u)$ 
from the domain $v\gg u$ to the domain $v\ll u$ one can rewrite the commutation 
relation \r{Ac1} in the form 
\begin{equation}\label{Ac2}
F_2(v)F_1(u)=f(v,u)F_1(u)F_2(v)+c\delta(u,v)F_{3,1}(u), 
\end{equation}
where the rational function $f(v,u)$ in the l.h.s. of \r{Ac2} should be understood as 
a series with respect to the nonnegative powers of $v/u$. 
The element 
$F_{3,1}(u)$ is the composed current which is equal to the residue 
of the product $F_2(v)F_1(u)$ at the point $v=u$
\begin{equation}\label{Ac3}
\begin{split}
&c\ F_{3,1}(u)=\mathop{\rm res}\limits_{v=u}F_2(v)F_1(u)=
(v-u)F_2(v)F_1(u)\Big|_{v=u}=\\
&\qquad=(v-u+c)F_1(u)F_2(v)\Big|_{v=u}=c\ F_1(u)F_2(u)\,.
\end{split}
\end{equation}

Let us demonstrate that 
the composed current given by the definition \r{Ac3}  
is a well-defined object in the 
completion  of the subalgebra $\oY_F$ compatible with the 
category of the highest weight  representations. 
This can be achieved by the normal ordering procedure in the 
product $F_2(v)\,F_1(u)$ described below. 
In the analytical language the definition \r{Ac3} means  \begin{equation}\label{an1}
\begin{split}
c\,F_{3,1}(v)&=\mathop{\rm res}\limits_{z=v}\,F_2(z)\,F_1(v)=-
\mathop{\rm res}\limits_{z=v}\,F_2(v)\,F_1(z)\\
&=\oint_v dz\, F_2(z)\,F_1(v)= - \oint_v dz\,F_2(v)\,F_1(z)
\end{split}
\end{equation}
or 
\begin{equation}\label{an2}
\begin{split}
c\,F_{3,1}(v)&=\oint_{C_\infty} 
dz\ F_2(z)\,F_1(v)- \oint_{C_0} dz\ 
\frac{1-(z+c)/v}{1-z/v}\ F_1(v)\,F_2(z)\,,\\
c\,F_{3,1}(v)&=\oint_{C_0} dz\ F_2(v)\,F_1(z)- \oint_{C_\infty} 
dz\ \frac{1-(v+c)/z}{1-v/z}\ F_1(z)\,F_2(v)\,,
\end{split}
\end{equation}
where $C_0$ is a contour around zero such that the point $z=v$ is outside 
the contour while $C_\infty$ is a contour close to infinity, which includes 
zero and the point $z=v$. 
One gets from \r{an2} 
\begin{equation}\label{an3}
\begin{split}
c\,F_{3,1}(v)&= [F_2[0],F_1(v)]+ c\,F_1(v)F^{+}_{2}(v)\,,\\
c\,F_{3,1}(v)&= [F_2(v),F_1[0]]- c\,F^{-}_{1}(v)\,F_2(v)\,,
\end{split}
\end{equation}
where for $i=1,2$
\begin{equation}\label{an4}
\begin{split}
F^{+}_{i}(v)&=\oint_{C_0} \frac{dz}{v}\,\frac{1}{1-z/v}\,F_i(z)=
\sum_{m\geq0}F_i[m]\,(v/c)^{-m-1}\,,\\
F^{-}_{i}(v)&=-\ \oint_{C_\infty} \frac{dz}{z}\, \frac{1}{1-v/z}\,F_i(z)
=-\sum_{m<0}F_i[m]\,(v/c)^{-m-1}
\end{split}
\end{equation}
are half-currents  and 
\begin{equation}\label{an5}
F_i(v)=F_{i}(v)^{+}-F_{i}(v)^{-}\,.
\end{equation}

By definition of the completed subalgebra $\oY_\aF$ of $\aDY{\fs\fl_3}$
each element of these subalgebra should be a linear span of 
monomials such that the half-currents $F_i^{-}(u)$ are on the 
left of the half-currents $F_i^{+}(v)$. It is seen from the formulas 
\r{an3} that the composed current $F_{3,1}(v)\in\oY_\aF$ since 
obviously 
\begin{equation*}
[F_2[0],F_1(v)]^{\pm}=[F_2[0],F^{\pm}_1(v)]\quad\mbox{and}\quad 
[F_2(v),F_1[0]]^{\pm}=[F^{\pm}_2(v),F_1[0]]\,.
\end{equation*}

In order to reveal the analytical properties of the 
product  $F_1(u)\,F_2(v)$ in the region $|u|\gg |v|$
we rewrite this product as sum of four terms  according to \r{an5} 
\begin{equation}\label{an8}
F_1(u)\,F_2(v)=F^{+}_{1}(u)\,F^{+}_{2}(v)+F^{-}_{1}(u)\,F^{-}_{2}(v)-
F^{-}_{1}(u)\,F^{+}_{2}(v)-F^{+}_{1}(u)\,F^{-}_{2}(v)\,.
\end{equation}
The first three terms are already normal ordered while the last one 
$F^{+}_{1}(u)\,F^{-}_{2}(v)$ can be ordered using the commutation 
relations between currents \r{Ac1} $(w-z+c)\,F_1(z)\,F_2(w)=(w-z)\,F_2(w)\,F_1(z)$.
Applying to this commutation relation the integral transform 
\begin{equation*}
-{\oint_{C_0}}\  \frac{dz}{u(1-z/u)}\ {\oint_{C_\infty}} \frac{dw}{w(1-v/w)}
\end{equation*}
and using the trivial identity 
\begin{equation*}
\frac{w-z+c}{(u-z)(w-v)}=\frac{v-u+c}{(u-z)(w-v)}+\frac{1}{u-z}+\frac{1}{w-v}
\end{equation*}
one gets using \r{an4}
\begin{equation}\label{an9}
(v-u+c)\,F^{+}_{1}(u)\,F^{-}_{2}(v)=
(v-u)\,F^{-}_{2}(v)\,F^{+}_{1}(u)+[F^{-}_{2}(v),F_1[0]]+
[F^{+}_{1}(u),F_2[0]]\,.
\end{equation}
This equality together with the formulas \r{an3} 
allows to present the product 
 $F_1(u)\,F_2(v)$ in the normal ordered form 
\begin{equation}\label{an10}
\begin{split}
F_1(u)\,F_2(v)&=F^{+}_{1}(u)\,F^{+}_{2}(v)+F^{-}_{1}(u)\,F^{-}_{2}(v)-
F^{-}_{1}(u)\,F^{+}_{2}(v)-\\
&\quad-\frac{v-u}{v-u+c}\,F^{-}_{2}(v)\,F^{+}_{1}(u)+\\
&\quad+\frac{c}{v-u+c}\Big(F^{+}_{3,1}(u)-F^{+}_{1}(u)\,F^{+}_{2}(u)+
\Big(F^{+}_{1}(u)\,F^{+}_{2}(u)\Big)^{+}\Big)-\\
&\quad-\frac{c}{v-u+c}\Big(F^{-}_{3,1}(v)+F^{-}_{1}(v)\,F^{-}_{2}(v)-
\Big(F^{-}_{1}(v)\,F^{+}_{2}(v)\Big)^{-}\Big).
\end{split}
\end{equation}
It explicitly shows the presence of the first order pole at $u=v+c$ 
in this product. Note that at $u=v$ the equality \r{an10} 
becomes a tautological identity. Repeating the same calculation and 
normal ordering, the product of the currents $F_2(v)\,F_1(u)$,
viewed as an element of the completed subalgebra $\oY_\aF$, 
will have a simple pole at $v=u$. 

One can repeat the same normal ordering calculations 
of the product of the same root currents using the commutation relations 
\begin{equation*}
f(u,v)\,F_1(u)\,F_1(v)=f(v,u)\,F_1(v)\,F_1(u)
\end{equation*}
to find that the product  $F_1(u)\,F_1(v)$ 
being rewritten in the normal ordered form 
reveals explicitly one zero when $u=v$ and one pole when $u=v-c$
\begin{equation*}
\begin{split}
&F_1(u)\,F_1(v)=\frac{u-v}{u-v+c}\Big(F^{+}_{1}(u)\,F^{+}_{1}(v)+
F^{-}_{1}(u)\,F^{-}_{1}(v)-F^{-}_{1}(u)\,F^{+}_{1}(v)+\\
&\qquad-F^{-}_{1}(v)\,F^{+}_{1}(u)+F^{-}_{1}(u,v)\,F^{-}_{1}(v)-
F^{+}_{1}(u)\,F^{+}_{1}(u,v)-F^{-,+}_{1}(u,v)\Big)\,,
\end{split}
\end{equation*}
where the combinations of half-currents 
\begin{equation*}
F^{\pm}_{1}(u,v)=\frac{F^{\pm}_{1}(u)-F^{\pm}_{1}(v)}{u-v}\,,\qquad
F^{-,+}_{1}(u,v)=\frac{F^{-}_{1}(u)\,F^{+}_{1}(v)-
F^{-}_{1}(v)\,F^{+}_{1}(u)}{u-v}
\end{equation*}
are regular at $u=v$ and normal ordered.

The same approach should be applied to the commutation relations
of the currents in the Yangian doubles $\aDY{\fg}$ for any $\ggo$. In order to equate any 
two expressions which contain  products of  currents one has 
to fix first a category of the highest weight representations where these products are considered and then 
normal order these products as it was explained above.

\section{Composed current commutation relation for $\aDY{\fg_2}$}\label{AppB}

In this appendix we will consider one particular commutation relation 
\r{sm-cr-g2} between simple root current $F_2(v)$ and the composed 
currents $F_{3,-1}(u)=-F_2(u-2\,c)\,F_2(u-3\,c)\,F_1(u)\,F_2(u)$. 
This composed current appear in 
the Yangian double $\aDY{\fg_2}$ which we will consider elsewhere.  
We will use the notation $  \ftr(u,v)$ for the rational function 
\begin{equation}\label{ftr}
    \ftr(u,v)=\frac{u-v+3\,c}{u-v}\,.
\end{equation}

To show that the Serre relation \r{Ser4} in the Yangian double $\aDY{\fg_2}$
implies the commutation relation \r{sm-cr-g2} we will define the composed 
currents  $F_{3,1}(u)=F_1(u)\,F_2(u)$ and 
$F_{3,0}(u)=-F_2(u-3\,c)\,F_1(u)\,F_2(u)$ inductively starting 
from the commutation relation 
\begin{equation*}
(v-u+3\,c)\,F_1(u)\,F_2(v)=(v-u)\,F_2(v)\,F_1(u)\,.
\end{equation*}
As it was explained in the previous appendix this relation can be rewritten in the form 
\begin{equation}\label{AB3}
F_2(v)\,F_1(u)=\frac{1}{\ftr(u,v+3\,c)}\,F_1(u)\,F_2(v)+3\,c\,\delta(u,v)\,F_{3,1}(u)\,,
\end{equation}
where the coefficient $\ftr(u,v+3\,c)^{-1}$ is a series $1-3\,c\,\sum_{\ell\geq0}v^\ell\,u^{-\ell-1}$.
This defines the composed current $F_{3,1}(u)=F_1(u)\,F_2(u)$ which is a well-defined 
object in the completed subalgebra $\oY_F$ of the Yangian double 
$\aDY{\fg_2}$ restricted to the category of the highest weight representations. 
Due to \r{FFalt}
the composed current $F_{3,1}(u)$ has the following commutation relation with the simple root 
current $F_2(v)$
\begin{equation}\label{AB6c}
(v-u+3\,c)\,f(u,v)\,F_{3,1}(u)\,F_2(v)=(v-u)\,f(v,u)\,F_2(v)\,F_{3,1}(u)\,.
\end{equation}
It can be rewritten in the analytical form  
\begin{equation}\label{AB6a}
f(u,v)\,F_{3,1}(u)\,F_2(v)=\frac{f(v,u)}{\ftr(v,u)}\,F_2(v)\,F_{3,1}(u)+2\,c\,\delta(u,v+3\,c)\,F_{3,0}(u)\,,
\end{equation}
which defines the composed currents $F_{3,0}(u)=-F_2(u-3\,c)\,F_1(u)\,F_2(u)$. 
The commutation relation of this composed current again with the simple root 
current $F_2(v)$ is 
\begin{equation*}
(v-u+3\,c)\,f(u-3\,c,v)\,f(u,v)\,F_{3,0}(u)\,F_2(v)=(v-u)\,f(v,u)\,f(v,u-3\,c)\,F_2(v)\,F_{3,0}(u)
\end{equation*}
which is equivalent to 
\begin{equation}\label{AB6d}
(v-u+2\,c)\,f(u,v)\,F_{3,0}(u)\,F_2(v)=(v-u+c)\,f(v,u-3\,c)\,F_2(v)\,F_{3,0}(u)\,.
\end{equation}
The latter equality in its analytical form 
\begin{equation}\label{AB6b}
f(u,v)\,F_{3,0}(u)\,F_2(v)=\frac{f(v,u-3\,c)}{f(v,u-c)}
\,F_2(v)\,F_{3,0}(u)+2\,c\,\delta(u,v+2\,c)\,F_{3,-1}(u)
\end{equation}
defines the composed current $F_{3,-1}(u)=-F_2(u-2\,c)\,F_2(u-3\,c)\,F_1(u)\,F_2(u)$. 
Note that the commutation relations \r{AB6c} and \r{AB6d} signify that 
the normalized 
products $f(u,v)\,F_{3,1}(u)\,F_2(v)$ and $f(u,v)\,F_{3,0}(u)\,F_2(v)$ being considered as 
an operator valued functions of the complex parameter $u-v$ have only  simple 
poles at $u-v=3\,c$ and $u-v=2\,c$ respectively. On the other hand the unnormalized 
products of the currents $F_{3,1}(u)\,F_2(v)$ and $F_{3,0}(u)\,F_2(v)$ have an
additional pole at $u-v=-c$ and a zero at $u=v$.

We rewrite  the Serre relation \r{Ser4}  using symmetrizations 
of the products of simple root currents 
\begin{equation*}
\cX_b(u;\bv)=F_2(v_1)\cdots F_2(v_{b})\,F_1(u)\,F_2(v_{b+1})\cdots F_2(v_4),\qquad
b=0,1,2,3,4
\end{equation*}
as follows
\begin{equation}\label{AB4}
\sum_{b=0}^4
(-1)^b\,\bfC^4_b\,\mathop{\rm Sym}\limits_{\bv}\sk{\cX_b(u;\bv)}=0 \,,
\end{equation}
where $\bv=\{v_1,v_2,v_3,v_4\}$ is a set of parameters and 
$\bfC^4_b=\frac{4!}{b!(4-b)!}$ is a binomial coefficient.

The commutation relations \r{AB3} allows to present each product  in the Serre relation \r{AB4} in the form 
\begin{equation}\label{AB7}
\cX_b(u;\bv)=3\,c\sum_{s=1}^b\delta(u,v_s)\,\cY_s(u;\bv_s)\prod_{a=s+1}^b\ftr(v_a,u)+
\cX_0(u;\bv)\prod_{s=1}^b\ftr(v_s,u)\,,
\end{equation}
where we denoted 
\begin{equation}\label{AB8}
\cY_s(u;\bv_s)=F_2(v_1)\cdots F_2(v_{s-1})\,F_{3,1}(u)\,
F_2(v_{s+1})\cdots F_2(v_4),\qquad s=1,2,3,4
\end{equation}
and assume, as usual, that $\sum_a^b(\cdot)=0$ and $\prod_a^b(\cdot)=1$ 
if $a<b$. In \r{AB7} and \r{AB8} the set $\bv_s$ means the set $\bar v$
 without the parameter $v_s$
 \begin{equation*}
 \bv_s=\bv\setminus\{v_s\}\,.
 \end{equation*}
Substituting expressions \r{AB7} in the Serre relation \r{AB4} we conclude that 
it is equivalent to 
\begin{equation}\label{AB9}
\mathop{\rm Sym}\limits_{\bv}\sk{3\,c
\sum_{s=1}^4\delta(u,v_s)\,\cY_s(u;\bv_s)\,\Ccc_s(u;\bv)+\cX_0(u;\bv)\,
\Ccc_0(u;\bv)}=0\,,
\end{equation}
where  
\begin{equation}\label{AB10}
\Ccc_s(u;\bv)=\sum_{b=s}^4(-1)^b\,\bfC^4_b
\prod_{a=s+1}^b\ftr(v_a,u),\quad s=0,1,2,3,4\,.
\end{equation}

Let us prove first that the second term in \r{AB9} is vanishing due to the 
commutation relation \r{FF}. Indeed,  using definition \r{ftr} one gets 
\begin{equation}\label{AB11}
\begin{split}
&\Ccc_0(u;\bv)=\sum_{b=0}^4
(-1)^b\,\bfC^4_b\prod_{s=1}^b\ftr(v_s,u)=c\,\prod_{s=1}^4
g(v_s,u)\times\\
&\times\Big(h(v_1,v_2)\,\mathcal{P}_{3,4}(u;v_3,v_4)-
2\,h(v_2,v_3)\,\mathcal{P}_{1,4}(u;v_1,v_4)+
h(v_3,v_4)\,\mathcal{P}_{1,2}(u;v_1,v_2)\Big)\,,
\end{split}
\end{equation}
where $\mathcal{P}_{a,b}(u;v_a,v_b)$ are quadratic polynomials 
in $u$, $v_a$, $v_b$, $c$ 
\begin{equation*}
\begin{split}
\mathcal{P}_{3,4}(u;v_3,v_4)&=u^2+u(v_4-3v_3)+v_3v_4+5cv_3-7cv_4+7c^2\,,\\
\mathcal{P}_{1,4}(u;v_1,v_4)&=u^2+u(v_1+v_4+3c)+v_1v_4-cv_1+4cv_4-4c^2\,,\\
\mathcal{P}_{1,2}(u;v_1,v_2)&=u^2+u(v_1-3v_2-6c)+v_1v_2+8cv_2+12c^2\,.
\end{split}
\end{equation*}

Due to the commutation relation in  \r{FF} 
which can be written in the form 
\begin{equation}\label{F2F2}
h(v_a,v_{a+1})\,F_2(v_a)\,F_2(v_{a+1})+h(v_{a+1},v_{a})\,F_2(v_{a+1})\,F_2(v_{a})=0
\end{equation}
we obtain that 
\begin{equation*}
\mathop{\rm Sym}\limits_{v_a,\ v_{a+1}}
\Big(F_1(u)\,F_2(v_1)\,F_2(v_2)\,F_2(v_3)\,F_2(v_4)\, h(v_a,v_{a+1})\Big)=0
\end{equation*}
for any $a=1,2,3$, 
hence proving that the second term in the Serre relation \r{AB9} 
is vanishing.  The Serre relation itself becomes 
\begin{equation}\label{AB12}
\mathop{\rm Sym}\limits_{\bv}
\sk{\delta(u,v_4)
\sum_{s=1}^4 \overline{\cY}_s(u;\bv_4)\ 
\overline{\Ccc}_s(u;\bv_4)}=0\,,
\end{equation}
where products of the currents $\overline{\cY}_s(u;\bv_4)$ are 
\begin{equation}\label{AB13}
\overline{\cY}_s(u;\bv_4)=F_2(v_1)\cdots F_2(v_{s-1})\,F_{3,1}(u)\,
F_2(v_{s})\cdots F_2(v_3)
\end{equation}
and 
\begin{equation}\label{AB14}
\overline{\Ccc}_s(u;\bv_4)=
\overline{\Ccc}_s(u;v_1,v_2,v_3)= \sum_{b=s}^4(-1)^b\,\bfC^4_b\prod_{a=s}^{b-1}
\ftr(v_a,u),\qquad s=1,2,3,4\,.
\end{equation}
Note that according to this definition $\overline{\Ccc}_4(u;\bv_4)=1$ and 
for $s=1,2,3$ the functions $\overline{\Ccc}_s(u;\bv_4)$ depends only 
on
$v_s-u,\ldots,v_3-u$.

Using the linear independence of the $\delta$-functions we conclude that 
the Serre relation \r{Ser4} implies the following relation between composed 
current $\aF_{3,1}(u)$ and the simple root currents $F_2(v_a)$ for 
$a=1,2,3$ (recall that the set 
$\bv_4=\{v_1,v_2,v_3\}$)
\begin{equation}\label{AB12a}
\mathop{\rm Sym}\limits_{\bv_4}
\sk{\sum_{s=1}^4 \overline{\cY}_s(u;\bv_4)\ 
\overline{\Ccc}_s(u;\bv_4)}=0\,.
\end{equation}

The l.h.s. of this equality is a linear combination of the product of the currents 
\r{AB13}. As it was explained at the beginning of this appendix 
these products have poles at $u=v_s+3\,c$,
$u=v_s-c$ and zeros at $u=v_s$
for $s=1,2,3$. In order to compensate the poles 
at $u=v_s-c$ and zeros at $u=v_s$
 we multiply relation \r{AB12a}  by the symmetric product 
$\prod_{s=1}^3f(u,v_s)$.
  Then the  equality \r{AB12a}
takes the form 
\begin{equation}\label{AB12b}
\mathop{\rm Sym}\limits_{\bv_4}
\sk{2\,c\,\sum_{s=1}^3\delta(u,v_s+3\,c)\mathcal{W}_s(u;\bv_4)\,\mathcal{R}_s(u;\bv_4)
+\overline{\cY}_4(u;\bv_4)\,\mathcal{R}_4(u;\bv_4)}=0\,,
\end{equation}
where we used the commutation relations  \r{AB6a}. In \r{AB12b} 
 $\mathcal{W}_s(u;\bv_4)$ are normalized  products of the currents 
\begin{equation*}
\mathcal{W}_s(u;\bv_4)=\prod_{a=s+1}^3f(u,v_a)\,F_2(v_1)\cdots F_2(v_{s-1})\,
F_{3,0}(u)\,F_2(v_{s+1})\cdots F_2(v_3)
\end{equation*}
and the rational functions $\mathcal{R}_s(u;\bv)$ are 
\begin{equation}\label{R-def2}
\mathcal{R}_s(u;\bv_4)=\sum_{b=1}^s\overline{\Ccc}_b(u;\bv)
\prod_{a=1}^{b-1}f(u,v_a)
\prod_{a=b}^{s-1}
\frac{f(v_a,u)}{\ftr(v_a,u)}\,.
\end{equation}

One can verify now that the rational function $\mathcal{R}_4(u;\bv_4)$
can be presented in the form 
\begin{equation*}
\mathcal{R}_4(u;\bv_4)=\prod_{s=1}^3g(u,v_s)\,g(v_s,u-3\,c)
\sum_{a=1}^2h(v_a,v_{a+1})\,\overline{\mathcal{P}}_a(u;\bv_4)\,,
\end{equation*}
where the non-factorizable polynomials  $\overline{\mathcal{P}}_a(u;\bv_4)$
are symmetric with respect to permutation of the parameters $v_a$ and $v_{a+1}$
for $a=1,2$. Then one has to consider the symmetrization over $\bar v_4$ of the product 
$\mathcal{W}_s(u;\bv_4)\sum_{a=1}^2h(v_a,v_{a+1})\,\overline{\mathcal{P}}_a(u;\bv_4)$.
Due to the symmetry property of the polynomials and to the commutation relations \r{F2F2}, this term vanishes.

Thus, we are led to
\begin{equation}\label{AB12c}
\mathop{\rm Sym}\limits_{\bv_4}
\sk{\delta(u,v_3+3\,c)\sum_{s=1}^3\overline{\mathcal{W}}_s(u;v_1,v_2)\,
\overline{\mathcal{R}}_s(u;v_1,v_2)}=0\,,
\end{equation}
where $\overline{\mathcal{W}}_s(u;v_1,v_2)$ are the normalized products of the currents 
\begin{equation*}
\overline{\mathcal{W}}_s(u;v_1,v_2)=\prod_{a=s}^2f(u,v_a)\,F_2(v_1)\cdots F_2(v_{s-1})\,
F_{3,0}(u)\,F_2(v_{s})\cdots F_2(v_2)
\end{equation*}
and the rational functions $\overline{\mathcal{R}}_s(u;v_1,v_2)$ are 
\begin{equation}\label{R-def1}
\begin{split}
\overline{\mathcal{R}}_3(u;v_1,v_2)&=\mathcal{R}_3(u;v_1,v_2,u-3\,c)\,,\\
\overline{\mathcal{R}}_2(u;v_1,v_2)&=\mathcal{R}_2(u;v_1,u-3\,c,v_2)=8\,\frac{v_1-u+4\,c}{v_1-u+3\,c}\,,\\
\overline{\mathcal{R}}_1(u;v_1,v_2)&=\mathcal{R}_1(u;u-3\,c,v_1,v_2)=-4\,,
\end{split}
\end{equation}
where $\overline{\mathcal{R}}_3(u;v_1,v_2)$ is a rational function that we refrain to write explicitly to lighten the presentation.
 Using again the linear independence of the $\delta$-functions we conclude that the 
relation \r{AB12c} is equivalent to
\begin{equation}\label{AB12d}
\mathop{\rm Sym}\limits_{v_1,v_2}
\sk{\sum_{s=1}^3\overline{\mathcal{W}}_s(u;v_1,v_2)\,
\overline{\mathcal{R}}_s(u;v_1,v_2)}=0\,.
\end{equation}

Denoting 
\begin{equation*}
\ff(v,u)=\frac{f(v,u-3\,c)}{f(v,u-c)}=\frac{(v-u+4\,c)(v-u+c)}{(v-u+3\,c)(v-u+2\,c)}
\end{equation*}
and using the commutation relation \r{AB6b} one can rewrite equality 
\r{AB12d} as follows
\begin{equation}\label{eq:ser-aux}
\begin{split}
&\mathop{\rm Sym}\limits_{v_1,v_2}
\Big(\overline{\mathcal{W}}_3(u;v_1,v_2)\,
\Big(\overline{\mathcal{R}}_3+
\overline{\mathcal{R}}_2\,\ff(v_2,u)+
\overline{\mathcal{R}}_1\,\ff(v_2,u)\,\ff(v_1,u)
\Big)\Big)+\\
&\qquad\qquad+2\,c\mathop{\rm Sym}\limits_{v_1,v_2}\Big(\delta(u,v_2+2\,c)
\Big(\overline{\mathcal{R}}_2 +
\overline{\mathcal{R}}_1 \,\ff(v_1,u)\Big)\,F_2(v_1)\,F_{3,-1}(u)+\\
&\qquad\qquad\qquad\qquad+\delta(u,v_1+2\,c)\,\overline{\mathcal{R}}_1 \,f(u,v_2)\,
F_{3,-1}(u)\,F_2(v_2)\Big)=0\,.
\end{split}
\end{equation}
Using the explicit form of the rational functions  \r{R-def1} and \r{R-def2} 
one can verify that the coefficient of $\overline{\mathcal{W}}_3(u;v_1,v_2)$ in the  
relation \r{eq:ser-aux} is equal to $h(v_1,v_2)\,\mathcal{Q}(u;v_1,v_2)$ where 
$\mathcal{Q}(u;v_1,v_2)$ is a rational 
function, symmetric in the variables $v_1$ and $v_2$. Then, the first line in \r{eq:ser-aux} 
vanishes due to the commutation relations \r{F2F2}. 

Now, taking into account the relation
\begin{equation*}
\overline{\mathcal{R}}_2(u;v_1,v_2) +
\overline{\mathcal{R}}_1(u;v_1,v_2) \,\ff(v_1,u)=
4\,\frac{v_1-u+4\,c}{v_1-u+2\,c}=4\,
\frac{\ftr(v_1,u-c)}{f(v1,u-c)}
\end{equation*}
and the linear independence of the $\delta$-functions 
we conclude that the relation 
\r{AB12d} implies the commutation relation
\begin{equation*}\label{AB15}
f(u,v)\,F_{3,-1}(u)\,F_2(v) =
\frac{\ftr(v,u-c)}{f(v,u-c)}\,
F_2(v)\,F_{3,-1}(u)\,.
\end{equation*}
This finishes the proof of proposition~\ref{Enr5}.\qed

\section{Identities for distributions and Serre relations}\label{AppC}

In this appendix, we provide the sketch of an alternative proof of the proposition~\ref{Enr1}--\ref{Enr4}, closer to the original proof done in 
\cite{E} for quantum affine algebras.
The proof given in section~\ref{sect31}
 uses the notion of  composed currents and is more informative since,  
 in addition to the vanishing of the elements \r{Ser-ele}, it also provides 
the structure of these elements in the vicinity of the Serre strata. 

\begin{prop}\label{En-iden}
For any two nodes $a$ and $b$  of the Dynkin diagram for 
the Lie algebra $\ggo$ such that $a\not=b$, $\fb_{a,b}\not=0$, and $\fg\not=\fg_2$
there are two equivalent  $\delta$-function identities
(for shortness we denoted $\fm_{a,b}\equiv \fm=2,3$)
\begin{itemize}
\item for $a<b$
\begin{equation}\label{En2}
\begin{split}
&\sum_{\sigma\in S_\fm}\sum_{s=0}^{\fm}\Big[{\fm\atop s}\Big]
\prod_{\ell_1<\ell_2}^{\fm}f_a(u_{\sigma(\ell_1)},u_{\sigma(\ell_2)})^{-1}
\prod_{\ell_1=1}^s h_{a,b}(u_{\sigma(\ell_1)},v)^{-1}\prod_{\ell_2=s+1}^{\fm}
g(v,u_{\sigma(\ell_2)})=\\
&=c^{\fm}\,
\sum_{\sigma\in S_\fm} \delta(v,u_{\sigma(1)})\,
\delta(u_{\sigma(1)},u_{\sigma(2)}-c\,\fb_{a,b} )\cdots
\delta(u_{\sigma(\fm-1)} ,u_{\sigma(\fm)}-c\,\fb_{a,b} )\,,
\end{split}
\end{equation}
\item for $a>b$
\begin{equation}\label{En3}
\begin{split}
&\sum_{\sigma\in S_\fm}\sum_{s=0}^{\fm}\Big[{\fm\atop s}\Big]
\prod_{\ell_1<\ell_2}^{\fm}f_a(u_{\sigma(\ell_1)},u_{\sigma(\ell_2)})^{-1}
\prod_{\ell_1=1}^s g(u_{\sigma(\ell_1)} ,v)\prod_{\ell_2=s+1}^{\fm}
h_{a,b}(v,u_{\sigma(\ell_2)})^{-1}=\\
&=c^{\fm}\,
\sum_{\sigma\in S_\fm} \delta(v,u_{\sigma(1)})\,
\delta(u_{\sigma(1)} ,u_{\sigma(2)}+c\,\fb_{a,b} )\cdots
\delta(u_{\sigma(\fm-1)} ,u_{\sigma(\fm)}+c\,\fb_{a,b} )\,.
\end{split}
\end{equation}
\end{itemize}
\end{prop}

The identity \r{En3} can be obtained from \r{En2} by renaming 
$a\leftrightarrow b$ and an overall shift of the parameters $u_1,\ldots,u_m$. 
The proof of \r{En2} or \r{En3} 
is the same as in \cite{E} and we skip it. 
Instead, we describe how the rational analog of the $\delta$-function identities invented 
in \cite{E} can be inferred from the Serre relations. 
Let us emphasize that the following is not a proof of these identities (which can be established by arguments very similar to those used in \cite{E}). 
Rather, it provides a recipe for deriving such possible identities directly from the Serre relations.

We first rewrite the Serre relations \r{SR-gF} for the simple root 
currents $F_a(u)$, $a\in\PSR_\ggo$ in the following form 
\begin{equation}\label{En1}
\sum_{\sigma\in S_\fm}
\sum_{s=0}^{\fm}(-1)^s\,\Big[{\fm\atop s}\Big]\,
\prod_{\ell_1=1}^s F_a(u_{\sigma(\ell_1)})\  F_b(v)\,
\prod_{\ell_2=s+1}^{\fm}F_a(u_{\sigma(\ell_2)})=0\,.
\end{equation}

In order to get the desired $\delta$-function identities one has 
to associate to each product of pair of  currents 
in \r{En1} the inverse function which is in the commutation relations 
\r{FFalt} of this pair of the currents. The right-hand-side is the product of $\delta$-terms corresponding to all poles appearing in the left-hand-side. 
For example, for $\fm=2$, to the product 
 $F_a(u_j)\, F_b(v)$ we associate the function 
$\ocre_{a,b}(u_j,v)^{-1}$, to the product $F_b(v)\, F_a(u_j)$
the function $\ocre_{b,a}(v,u_j)^{-1}$ ($j=1,2$) and finally to the product 
$F_a(u_1)\, F_a(u_2)$ the function $\ocre_{a,a}(u_1,u_2)^{-1}$. Supposing $a<b$, the corresponding right-hand-side is $\sum_\sigma\delta(v,u_{\sigma(1)})\,\delta(u_{\sigma(1)},u_{\sigma(2)}-c\,\fb_{a,b} )$.

The functions which we associate to the products of  currents and 
taken from the commutation relations \r{FF}
  should be always understood  as a series 
in the domain where the modulus of the ratio of the second argument to the 
first one is less than 1. 

Using the definition \r{CuCu-alt} of the functions $\ocre_{a,b}(u,v)$  
 we are led to the proposition \ref{En-iden}, which is a rational analog of the 
proposition~2.1 in the paper \cite{E}. \qed

We now demonstrate how these identities lead to the vanishing of the 
elements $\sF_{a,b}(\bt^a,\bt^b)$ defined by the equalities \r{Ser-ele}. 
The identity \r{En2} can be used to prove the vanishing of the elements 
\r{pr-cu-1} and \r{pr-cu-13}  of the proposition~\ref{Enr1} and \ref{Enr3}.
Correspondingly, the identity \r{En3} is more convenient to 
prove the vanishing of the elements \r{pr-cu-2}  and \r{pr-cu-14} 
of the propositions~\ref{Enr2} and \ref{Enr4}. 
The r.h.s. of identites \r{En2} and \r{En3} is understood as the symmetrization 
of the product of two (resp. three) $\delta$-functions, when $\fm=2$ (resp. $\fm=3$).

Let us demonstrate how the identity \r{En1} for $\fm=2$,  $b=a+1$
and $\fb_{a,a+1}=-1/2$ 
provides the vanishing of the element \r{pr-cu-1} 
in  the proposition~\ref{Enr1}. 
From the expression \r{pr-cu-1}, we get
\begin{equation}\label{A9}
\begin{split}
\sF_{a,a+1}(\{u_1,u_2\},v)&=\frac{(u_2-u_1+c)}{(u_2-u_1)}\,\frac{(u_2-v-c)}{c}\,
\frac{(u_1-v-c)}{c}\ F_a(u_2)F_a(u_1)F_{a+1}(v)\,,\\
&=-\ \frac{(u_2-u_1+c)}{(u_2-u_1)}\,\frac{(u_2-v-c)}{c}\,
\frac{(v-u_1)}{c}\ F_a(u_2)F_{a+1}(v)F_a(u_1)\,,\\
&=\frac{(u_2-u_1+c)}{(u_2-u_1)}\,
\frac{(v-u_2)}{c}\,\frac{(v-u_1)}{c}\ F_{a+1}(v)F_a(u_2)F_a(u_1)\,,
\end{split}
\end{equation}
so that we can substitute
\begin{equation}\label{subs}
\begin{split}
F_a(u_2)F_a(u_1)F_{a+1}(v)&=\frac{u_2-u_1}{u_2-u_1+c}\ 
\frac{c^2\ \sF_{a,a+1}(\{u_1,u_2\},v)}{(u_1-v-c)(u_2-v-c)}\,,\\
F_a(u_2)F_{a+1}(v)F_a(u_1)&=- \frac{u_2-u_1}{u_2-u_1+c}\ 
\frac{c^2\ \sF_{a,a+1}(\{u_1,u_2\},v)}{(v-u_1)(u_2-v-c)}\,,\\
F_{a+1}(v)F_a(u_2)F_a(u_1)&=- \frac{u_2-u_1}{u_2-u_1+c}\ 
\frac{c^2\ \sF_{a,a+1}(\{u_1,u_2\},v)}{(v-u_1)(v-u_2)}
\end{split}
\end{equation}
into the Serre relation \r{En1}. Then, using the identity \r{En2}, it leads to the equality
\begin{equation}\label{A10}
0= c^2\ \mathop{\rm Sym}\limits_{u_1,u_2}\Big(\delta(v,u_1)\delta(v,u_2+c)\Big)\ \sF_{a,a+1}(\{u_1,u_2\},v)\,.
\end{equation}

The vanishing of the element $\sF_{a,a+1}(\{u_1,u_2\},v)$ at the 
Serre strata 
$v=u_1=u_2-c$
and $v=u_2=u_1-c$
follows from the  symmetry of the element 
$\sF_{a,a+1}(\{u_1,u_2\},v)=\sF_{a,a+1}(\{u_2,u_1\},v)$ 
and the linear independence of 
the $\delta$-functions.


\begin{thebibliography}{99}


\bibitem{AMR}
D.~Arnaudon, A.~Molev, E.~Ragoucy. 
{\sl On the R-matrix realization of Yangians and their representations.}
Ann. Henri Poincare, 7 (2006), 1269--1325.



\bibitem{DKh}
J.~Ding, S.~Khoroshkin. {\sl Weyl group extension of quantized current algebras.} 
Transform. Groups 5 (2000), 35--59. 

\bibitem{DKhP}
J.~Ding, S.~Khoroshkin, S.~Pakuliak. 
{\sl Integral Presentations for the Universal $\mathcal{R}$-Matrix}.
Lett Math Phys 53, 121--141 (2000), {\tt arXiv:math/0008226}, 
https://doi.org/10.1023/A:1026730817516 



\bibitem{D85}
V.~G.~Drinfeld, {\sl Hopf algebras and the quantum Yang-Baxter equation},
Sov.Math.Dokl. {\bf 32} (1985) 254--258.




\bibitem{D} V.~G.~Drinfel'd, {\sl Quantum groups},
J Sov Math, \textbf{41} (1988) 898.

\bibitem{Dnew} V.~G.~Drinfeld. \textsl{A new realization of 
Yangians and of quantum affine algebras}, Soviet Math. Dokl. {\bf 36}
(1988) 212--216.

\bibitem{D-FTINT} V.~G.~Drinfeld. \textsl{A new realization of 
Yangians and of quantum affine algebras}, Preprint FTINT 30--86, (1986)  (in Russian). 







\bibitem{E}
B.~Enriquez. {\sl On correlation functions of Drinfeld 
currents and shuffle algebras.} 
Transformation Groups 5 (2000), 111--120, {\tt arXiv:math/9809036},
https://doi.org/10.1007/BF01236465

\bibitem{EKhP} B.~Enriquez, S.~Khoroshkin, S.~Pakuliak, \textsl{Weight functions and Drinfeld currents,}
{Comm. Math. Phys.} {\bf 276} (2007) 691--725, \texttt{arXiv:0610398}.





\bibitem{HLPRS17} A.~Hutsalyuk,  A.~Liashyk, S.~Z.~Pakuliak,
E.~Ragoucy, N.~A.~Slavnov,
\textsl{Current presentation for the double super-Yangian
$DY(\mathfrak{gl}(m|n))$ and Bethe vectors},
Russ. Math. Surv. {\bf 72}:1  (2017) 33--99, \texttt{arXiv:1611.09620}.


\bibitem{KhT-DY}
S.~Khoroshkin,  V.~N.~Tolstoy.
{\sl Yangian double}, Lett. Math. Phys. {\bf 36} (1996), 373--402.

 \bibitem{KhP} 
S.~Khoroshkin, S.~Pakuliak, {\sl A computation of an universal weight function for
the quantum affine algebra $U_q(\mathfrak{gl}(N))$},  {J. of Mathematics of Kyoto University},
{\bf 48} n.2 (2008) 277--321, \texttt{arXiv:0711.2819}.


\bibitem{KT2} 
M.~Kosmakov, V.~Tarasov, {\sl
New combinatorial formulae for nested Bethe vectors II},
Lett. Math. Phys. {\bf 115} (2025), 12, \texttt{arXiv:2402.15717}.


\bibitem{LP1}
A.~Liashyk, S.~Pakuliak.
{\sl Gauss Coordinates vs Currents for the Yangian Doubles of the Classical Types.}
SIGMA {\bf 16} (2020) 120, 23 pages, {\tt arXiv:2006.01579}.

\bibitem{LPR-RR}
A.~Liashyk, S.~Pakuliak, E.~Ragoucy.
{\sl Bethe Vectors in Quantum Integrable Models
with Classical Symmetries}, 39 pages, {\tt 	arXiv:2601.00713}.



\bibitem{LPR-BV-CU}
A.~Liashyk, S.~Pakuliak, E.~Ragoucy.
{\sl Yangian Doubles and off-Shell Bethe Vectors,} {\tt to appear.}


\bibitem{S22}
N.~Slavnov. {\sl Algebraic Bethe Ansatz and Correlation Functions. 
An Advanced Course}, World Scientific, (2022). 


\bibitem{TF79}
L.~A.~Takhtadzhan, L.~D.~Faddeev, {\sl Quantum method of the inverse 
problem and the Heisenberg $XYZ$-model}, Russ Math Surv {\bf 34} 5 (1979)
11--68. 



\bibitem{Yang} C.~N.~Yang. {\sl Some Exact Results for the Many-Body Problem
in one Dimension with Repulsive Delta-Function Interaction}, Phys. Rev.
Lett. (1967)  {\bf 19}  1312--1315.

\end{thebibliography}
\end{document}